\documentclass[sigconf, nonacm=true]{acmart}
\acmDOI{10.1145/3452021.3458308}

\usepackage{booktabs} %
\usepackage[utf8]{inputenc}
\usepackage[ruled,vlined,linesnumbered]{algorithm2e} %
\usepackage{xargs}

\usepackage{xspace}
\usepackage{bigmacro}

\title{Tractability Beyond $\beta$-Acyclicity for Conjunctive Queries with Negation}

\author{Matthias Lanzinger}
\orcid{0000-0002-7601-3727}
\affiliation{%
  \institution{University of Oxford}
  \country{United Kingdom}
}
\affiliation{%
  \institution{TU Wien}
  \country{Austria}
}
\email{matthias.lanzinger@cs.ox.ac.uk}



\begin{abstract}
Numerous fundamental database and reasoning problems are known to be \textsf{NP}-hard in general but tractable on instances where the underlying hypergraph structure is $\beta$-acyclic. Despite the importance of many of these problems, there has been little success in generalizing these results beyond acyclicity.

In this paper, we take on this challenge and propose \emph{nest-set width}, a novel generalization of hypergraph $\beta$-acyclicity.  We demonstrate that nest-set width has desirable properties and algorithmic significance. In particular, evaluation of boolean conjunctive queries with negation (\cqneg) is tractable for classes with bounded nest-set width. Furthermore, propositional satisfiability (\sat) is fixed-parameter tractable when parameterized by nest-set width.

 \end{abstract}

\begin{document}
\fancyhead{}
\maketitle

\section{Introduction}

Hypergraph cyclicity has been identified as a key factor for the computational complexity of multiple fundamental database and reasoning problems. 
While various natural notions of hypergraph acyclicity exist, the
two most general ones --- $\alpha$- and $\beta$-acyclicity ---  have proven to be the
most relevant in the study of the complexity of reasoning. Important problems that are \np-hard in general often become tractable when restricted to acyclic instances. An example from databases is the evaluation of conjunctive queries (CQs), which is \np-hard in general but becomes tractable when the underlying hypergraph structure of the query is $\alpha$-acyclic~\cite{DBLP:conf/vldb/Yannakakis81}. The restriction to $\beta$-acyclic instances yields tractable classes for a variety of fundamental problems, including \sat~\cite{DBLP:journals/tcs/OrdyniakPS13}, \#\sat~\cite{DBLP:conf/stacs/Brault-BaronCM15}, and \cqneg evaluation~\cite{DBLP:conf/csl/Brault-Baron12,DBLP:conf/pods/NgoNRR14}. Notably, these problems remain \np-hard  when restricted to $\alpha$-acyclic instances or \#\textsf{P}-hard in the case of \#\sat (see~\cite{DBLP:conf/sat/CapelliDM14}).

However, while both types of acyclicity have proven to be interesting, the generalization of
$\alpha$-acyclicity has received significantly more attention. There, a rich hierarchy of width measures has been developed over the last two decades.
The \emph{tree-likeness} of $\alpha$-acyclic hypergraphs has been successfully generalized by hypertree width~\cite{DBLP:journals/jcss/GottlobLS02} and even further to fractional hypertree width~\cite{2014grohemarx}. What makes these generalizations particularly interesting is that they remain sufficient conditions for tractable CQ evaluation, emphasizing the deep connection between cyclicity and the complexity of CQs. The yet more general submodular width~\cite{DBLP:journals/jacm/Marx13}  characterizes the fixed-parameter tractability of CQ evaluation on the hypergraph level.
Moreover, related research has revealed notable parallels to other fields, e.g., to game theory~\cite{DBLP:journals/jair/GottlobGS05} and information theory~\cite{DBLP:journals/siamcomp/AtseriasGM13,panda}.

Despite the unquestionable success of the generalization of $\alpha$-acyclicity,  the generalization of $\beta$-acyclicity has received little attention so far. In the most prominent approach, Gottlob and Pichler~\cite{DBLP:conf/icalp/GottlobP01} introduced $\beta$-hypertree width (\bhw) as an analogue to hypertree width. In particular, they define \bhw as the maximum hypertree width over all subhypergraphs, mirroring a characterization of $\beta$-acyclicity in terms of every subhypergraph being $\alpha$-acyclic. However, it is difficult to exploit low \bhw algorithmically. An inherent problem with \bhw is that a witness for low \bhw would need to include a hypertree decomposition for each of the, exponentially many, subhypergraphs. Furthermore, none of the problems listed above as tractable on $\beta$-acyclic instances are known to be tractable for bounded \bhw (beyond those that are tractable for the more general bounded $hw$).

In recent work, Carbonell, Romero, and Zivny introduced point-decompositions and the accompanying point-width ($pw$)~\cite{DBLP:conf/lics/Carbonnel0Z19}, which generalizes both $\beta$-acyclicity and MIM-width\cite{DBLP:journals/jair/SaetherTV15}. They show that, given a point-decomposition of bounded point-width and polynomial size, \textsc{Max-CSP} can be decided in polynomial time. However, just as with \bhw, it is not known if $pw \leq k$ can be decided in polynomial time, even for constant $k$.
In summary, neither of these approaches allows us to extend the tractability under $\beta$-acyclicity for the problems mentioned above to larger tractable fragments. Considering the importance of the affected problems and the restrictiveness of $\beta$-acyclicity the situation is unsatisfactory for theoretical and practical use.

In this paper, we propose a new generalization of $\beta$-acyclicity
which we call nest-set width ($nsw$). In contrast to \bhw and $pw$, it
is not based on decompositions but instead generalizes a
characterization of $\beta$-acyclicity by the existence of certain kinds of elimination
orders. Nest-set width has several attractive properties that
suggest it to be a natural extension of
$\beta$-acyclicity. Importantly, $nsw \leq k$ can be decided in fixed-parameter tractable time when parameterized by $k$. Furthermore, we show that bounded $nsw$ yields new islands of tractability for \sat and
\cqneg evaluation. The full contributions of this paper are summarized as follows:
\begin{itemize}
\item We introduce a new hypergraph width notion --  nest-set width -- that generalizes the existence of nest point elimination orders.
\item We establish the relationship of $\nsw$ to other related widths. In particular, we show that bounded $\nsw$ is a special case of bounded \bhw and incomparable to other prevalent width measures such as bounded clique width and treewidth.
\item It is shown that deciding $\nsw \leq k$ is \np-complete when $k$ is part of the input but fixed-parameter tractable when parameterized by $k$.
\item Building on work by Brault-Baron~\cite{DBLP:conf/csl/Brault-Baron12} for the $\beta$-acyclic case, we show the tractability of  evaluation of boolean CQs with negation for classes with bounded $\nsw$.
\item Finally, we demonstrate how to derive the fixed-parameter tractability of \sat parameterized by $\nsw$
  from our main result.
  
\end{itemize}

The rest of the paper is structured as follows. Section~\ref{sec:prelim} introduces necessary notation and preliminaries. We define nest-set width and establish some basic properties in Section~\ref{sec:nsw}. We move on to establish the relationship between $\nsw$ and other width measures, most importantly \bhw, in Section~\ref{sec:betahw}. The complexity of checking $nsw$ is discussed in Section~\ref{sec:complexity}.
The tractability of \cqneg under bounded $\nsw$ is shown in Section~\ref{sec:scq}. Concluding remarks and a discussion of future work in Section~\ref{sec:conclusion} complete the paper. %

\section{Preliminaries}
\label{sec:prelim}

For positive integers $n$ we will use $[n]$ as a shorthand for the set $\{1,2,\dots,n\}$. When $X$ is a set of sets
we sometimes write $\bigcup X$ for $\bigcup_{x\in X}x$. The same applies analogously to intersections.

Linear orders will play an important role throughout this paper. Recall, a binary relation $R$ is a \emph{linear order} if it is antisymmetric, transitive and connex (either $aRb$ or $bRa$ for all $a$ and $b$). We will be particularly interested in whether the subset relation $\subseteq$ is a linear order on some domain. If $\subseteq$ is a linear order for some set $X$, we say $X$ is \emph{linearly ordered} by $\subseteq$. Note that $\subseteq$ is inherently transitive and antisymmetric and we can limit our arguments to connexity.

We use standard notions from (parameterized) computational complexity
theory such as reductions and the classes \textsf{P} and \np. We refer
to~\cite{DBLP:books/daglib/0018514}
and~\cite{DBLP:books/sp/CyganFKLMPPS15} for comprehensive overviews of
computational complexity and parameterized complexity, respectively.
Furhtermore, we assume the reader to be familiar with propositional logic.

\subsection{Hypergraphs, Acyclicity \& Width}
\label{sec:hgprelim}

A \emph{hypergraph} $H$ is a pair $(V(H), E(H))$ where $V(H)$ is a set of \emph{vertices} and $E(H) \subseteq 2^{V(H)}$
is a set of \emph{hyperedges}.
For hypergraph $H$ and vertex $v$, we denote the set of incident edges of $v$ as $\neigh(v, H) := \{ e\in E(H) \mid v \in e\}$.
The notation is extended to sets of vertices $s = \{v_1, \dots, v_\ell\}$ as $\neigh(s, H) := \bigcup_{i=1}^\ell \neigh(v_i, H)$.
We say an edge $e \in \neigh(s, H)$ is \emph{incident} to the set $s$. If $H$ is clear from the context we drop $H$ in the argument and write only $\neigh(s)$.

A \emph{subhypergraph} $H'$ of $H$ is a hypergraph with
$E(H') \subseteq E(H)$ and $V(H') = \bigcup E(H')$.  The \emph{vertex
  induced subhypergraph} $H[U]$ of $H$ is the hypergraph with
$V(H[U]) = U$ and
$E(H[U]) = \{e\cap U\mid e \in E(H)\} \setminus \{\emptyset\}$.  For a set
of vertices $X$ we write $H-X$ as shorthand for the vertex induced
subhypergraph $H[V(H)\setminus X]$.

All common notions of hypertree acyclicity have numerous equivalent definitions (see e.g.,~\cite{fagin1983degrees,DBLP:journals/csur/Brault-Baron16}). Here we recall only those definitions that are necessary to present the results of this paper.

A \emph{join tree} of $H$ is a pair $(T, \epsilon)$ where $T$ is a tree and $\epsilon: T \to E(H)$ is a bijection from the nodes in $T$ to the edges of $H$ such that the following holds: for every $v \in V(H)$ the set $\{ u\in T\mid v \in \epsilon(u)\}$ is a subtree of $T$. If $H$ has a join tree, then we say that $H$ is \emph{$\alpha$-acyclic}.

A (weak) \emph{$\beta$-cycle} is a sequence $(e_1, v_1, e_2,\dots, v_{n-1}, e_n, v_n, e_{n+1})$ with
$n \geq 3$ where $e_1,\dots,e_n$ are distinct hyperedges, $e_1=e_{n+1}$, and $v_1,\dots,v_n$ are
distinct vertices. Moreover, for all $i \in [n]$, $v_i$ is in
$e_i$ and $e_{i+1}$ and not in any other edge of the sequence. A hypergraph is \emph{$\beta$-acyclic}
if it has no $\beta$-cycle.

An alternative (equivalent) definition of $\beta$-acyclicity is that $H$ is $\beta$-acyclic if and only if all subhypergraphs of $H$ are $\alpha$-acyclic. In this paper, a third characterization of $\beta$-acyclicity will be important.
We call a vertex $v$ of $H$ a \emph{nest-point} if $I(v)$ is linearly ordered by $\subseteq$. We can then characterize $\beta$-acyclicity by a kind of elimination order for nest-points (this will be made more precise for a more general case in Definition~\ref{def:neo}).

\begin{proposition}[\cite{DBLP:journals/ipl/Duris12}]
  \label{prop:acyc.elim}
  A hypergraph $H$ is $\beta$-acyclic if and only if the empty hypergraph can be reached by successive removal of nest-points and empty-edges from $H$.
\end{proposition}

Join trees have been successfully generalized to hypertree decompositions.
A \emph{hypertree decomposition}~\cite{DBLP:journals/jcss/GottlobLS02} of a hypergraph $H$ is a tuple $\hdecomp$, where $T$ is a rooted tree, for every node $u$ of the tree, $B_u \subseteq V(H)$ is called the \emph{bag} of node $u$, and $\lambda_u \subseteq E(H)$ is the \emph{cover} of $u$.
Furthermore, $\hdecomp$ must satisfy the following properties.
\begin{enumerate}
\item The subgraph $T_v  =\{u \in T \mid v \in B_u\}$ for vertex $v \in V(H)$ is a tree.
\item For every $e \in E(H)$ there exists a $u \in T$ such that $e \subseteq B_u$.
\item For every node $u$ in $T$ it holds that $B_u \subseteq \bigcup \lambda_u$.
\item Let $T_u$ be the subtree of $T$ rooted at node $u$ and let $B(T_u)$ be the union of all bags of nodes in $T_u$.
  For every node $u$ in $T$ it holds that $\bigcup \lambda_u \cap B(T_u) \subseteq B_u$.
\end{enumerate}
The first property is commonly referred to as the \emph{connectedness condition} and the fourth property is called the \emph{special condition}. The \emph{hypertree width} ($\hw$) of a hypertree decomposition is $\max_{u\in T}(|\lambda_u|)$ and the hypertree width of $H$ ($\hw(H)$) is the minimal width of all hypertree decompositions of $H$.

If we exclude the special condition in the above list of properties, we obtain the definition of a \emph{generalized hypertree decomposition}. The \emph{generalized hypertree width} of hypergraph $H$ ($\ghw(H)$) is defined analogously to before as the minimal width of all generalized hypertree decompositions of $H$.

It is known that $hw(H)=1$ if and only if $H$ is
$\alpha$-acyclic~\cite{DBLP:journals/jcss/GottlobLS02}. Analogous to
the definition of $\beta$-acyclicity in terms of every subhypergraph
being $\alpha$-acyclic, Gottlob and
Pichler~\cite{DBLP:conf/icalp/GottlobP01} introduced
\emph{$\beta$-hypertree width} \bhw$(H) = \max \{ \hw(H') \mid H'$ is a subhypergraph of $ H\}$. Note that
we therefore also have \bhw$(H) = 1$ if and only if $H$ is
$\beta$-acyclic.

We will make some comparisons to some further well-known width notions of hypergraphs: treewidth and clique width of the primal and incidence graph. The technical details of these concepts are of no importance in this paper and we refer to~\cite{DBLP:conf/icalp/GottlobP01} for full definitions.

\subsection{Conjunctive Queries}
\label{sec:cqprelim}

A \emph{signature} $\sigma$ is a finite set of relation symbols with associated arities. We write $ar(R)$ for the arity of relation symbol $R$. A \emph{database} $D$ (over signature $\sigma$) consists of a \emph{finite} domain $Dom$ and a relation $R^D$ for each relation symbol $R$ in the signature.

A \emph{conjunctive query with negation} (over signature $\sigma$) is a set of literals. A literal is
of the form $L(v_1, \dots, v_m)$ where $v_1,\dots,v_m$ are variables and $L$ is either $R$ or $\neg R$ for any $m$-ary relation symbol $R$ in $\sigma$. If $L$ is of the form $R$ we call the literal \emph{positive}, otherwise, if it is of the form $\neg R$ we say that the literal is \emph{negative}.
We commonly refer to a \cqneg simply as \emph{query}. We write $vars(q)$ for the set of all variables that occur in the literals of query $q$. We sometimes denote queries like logical formulas, i.e., $R_1(\vec{v_1}) \land \cdots  \land R_n(\vec{v_n})$ with the understanding that the query is simply the set of all conjuncts.

Let $q$ be a query and $D$ a database over the same signature. We call a function $a \colon vars(q) \to Dom$ an \emph{assignment} for $q$. For a set of variables $X$ we write $a[X]$ for the assignment with domain restricted to $X$. In a slight abuse of notation we also write $a[\vec{v}]$ for the tuple $(a(v_1),\dots,a(v_n))$ where $\vec{v}=(v_1,\dots,v_n)$ is a sequence of variables. An \emph{extension} of an assignment $a \colon \mathit{Vars} \to Dom$ is an assignment $a' \colon \mathit{Vars}' \to Dom$ with $\mathit{Vars}' \supset \mathit{Vars}$ and $a(v) = a'(v)$ for every variable $v \in \mathit{Vars}$.

We say that the assignment $a$ \emph{satisfies a positive literal} $R(\vec{v})$ if $a[\vec{v}]\in R^D$. Similarly, $a$ \emph{satisfies a negative literal} $\neg R(\vec{v})$ if $a[\vec{v}] \not \in R^D$. An assignment \emph{satisfies a query} $q$ (over database $D$) if it satisfies all literals of $q$. We write $q(D)$ for the set of all satisfying assignments for $q$ over $D$. We can now define the central decision problem of this paper.
\begin{problem}[framed]{\cqnegprob}
   Instance: & A \cqneg $q$ and a database $D$ \\
   Question: & $q(D) \neq \emptyset$?
\end{problem}

A query $q$ has an associated hypergraph $H(q)$. The
vertices of $H(q)$ are the variables of $q$. Furthermore, $H(q)$ has
an edge $\{v_1,\dots, v_n\}$ if and only if there exists a literal
$R(v_1,\dots,v_n)$ or $\neg R(v_1,\dots,v_n)$ in $q$.

To simplify later arguments we will assume that every relation symbol
occurs only once in a query. We will therefore sometimes write the
relation symbol, without the variables, to identify a literal.  Note
that every instance of \cqnegprob can be made to satisfy this
property, by copying and renaming relations, in linear time.

Finally, for our algorithmic considerations we assume a reasonable representation of queries and databases. In particular we assume that a relation $R$ has a representation of size $\norm{R}=O(|R|\cdot ar(R)\cdot \log Dom)$. Accordingly, we assume the \emph{size of a database} $D$ as $\norm{D} = \norm{Dom} + \sum_{R\in \sigma} \norm{R}$ and the \emph{size of a query} $q$ as $\norm{q} = O\left(\sum_{R \in \sigma} ar(R) \log|vars(q)|\right)$.
Finally, we refer to the cardinality of the largest relation in $D$ as $|R_{max}(D)|=\max_{R\in \sigma}|R^D|$. When the database is clear from the context we write just $|R_{max}|$.

\section{Nest-Set Width}
\label{sec:nsw}
In this section we introduce nest-set width and establish some of its basic
properties.  The crucial difference between \bhw and $\nsw$ is that the
generalization is based on a different characterization of
$\beta$-acyclicity. While \bhw generalizes the condition of every subgraph having a
join tree, nest-set width instead builds on the
characterization via nest point elimination from Proposition~\ref{prop:acyc.elim}. We start by generalizing nest points to \emph{nest-sets}:

\begin{definition}[Nest-Set]
  Let $H$ be a hypergraph. A non-empty set $s \subseteq V(H)$ of vertices is called
  a \emph{nest-set} in $H$ if the set
  \[
    \neight(s, H) := \{ e\setminus s  \mid e \in \neigh(s, H)\}
  \]
  is linearly ordered by $\subseteq$.
\end{definition}

As the comparability by $\subseteq$ of sets minus a nest-set will appear frequently,
 we introduce explicit notation for it.  Let $H$ be a
hypergraph and $s \subseteq V(H)$. For two sets of vertices
$V,U \subseteq V(H)$, we write $V \subseteq_s U$ for
$V \setminus s \subseteq U \setminus s$.  We could thus alternatively
define nest-sets as those sets $s$ for which $\neigh(s,H)$ is linearly
ordered by $\subseteq_s$.

In later sections, the maximal elements with respect to $\subseteq_s$
will play an important role. For a nest-set $s$ we will refer to a maximum
edge in $\neigh(s)$ w.r.t. $\subseteq_s$ as a \emph{guard} of $s$.
Note that there may be multiple guards. However, in all of the following usage
it will make no difference which guard is used and we will implicitly
always use the lexicographically first one (and thus refer to
\emph{the} guard).

Like for nest points, we want to investigate how a hypergraph can
be reduced to the empty hypergraph by successive removal of nest
sets. We formalize this notion in the form of \emph{nest-set elimination orderings}.

\begin{definition}[Nest-Set Elimination Ordering]
  \label{def:neo}
  Let $H$ be a hypergraph and let $\calO = (s_1, \dots, s_q)$ be a
  sequence of sets of vertices.  Define $H_0 = H$ and $H_i := H_{i-1}-s_i$.  We call $\calO$ a \emph{nest-set
    elimination ordering} (NEO) if, for each $i \in [q]$, $s_i$ is a
  nest-set of $H_{i-1}$ and $H_{q}$ is the empty hypergraph.
\end{definition}

Note that an elimination ordering is made up of at most $|V(H)|$
nest-sets. We are particularly interested in how large the nest-sets
have to be for a NEO to exist. Hence, we introduce notation for
restricted-size nest-sets and NEOs:

\begin{itemize}
\item If $s$ is a nest-set of $H$ with at most $k$ elements then we
  call $s$ a \emph{$k$-nest-set}.
\item A nest-set elimination ordering that consists of only
  $k$-nest-sets is a \emph{$k$-nest-set elimination ordering} ($k$-NEO).
\item Finally, the \emph{nest-set width} $\nsw(H)$ of a hypergraph $H$
  is the lowest $k$ for which there exists a $k$-NEO.
\end{itemize}

It is easy to see that a hypergraph has a $1$-nest-set $\{v\}$ if and only if $v$ is a nest point.
Therefore, a $1$-NEO corresponds directly to a sequence of nest point deletions that eventually result in the empty hypergraph. As this is exactly the characterization of $\beta$-acyclicity from Proposition~\ref{prop:acyc.elim}, we see that $\nsw$ generalizes $\beta$-acyclicity.

\begin{corollary}
  A hypergraph $H$ has $\nsw(H)=1$ if and only if $H$ is $\beta$-acyclic.
\end{corollary}

\begin{example}
  \label{ex:nsw}
  Let $H_0$ be the hypergraph with edges $\{a,b,c,d\}$, $\{a,d,e\}$, $\{c,d,f\}$, $\{b,e\}$, and $\{c,f\}$.
  Figure~\ref{fig:nswex} illustrates the step-wise elimination of $H_0$ according to the $2$-NEO
  $(\{c,f\}, \{b,e\}, \{a,d\})$.

  For the first nest-set $s_1=\{c,f\}$ we see that
  $\neigh(s_1, H_0)= \{\{a,b,c,d\}$, $\{c,d,f\}, \{c,f\}\}$ and
  $\neight(s_1,H_0) = \{\{a,b,d\}, \{d\}, \emptyset\}$.
  To verify that $s_1$ is a nest-set of $H_0$ we observe that
  $\{a,b,d\} \supseteq \{d\} \supseteq \emptyset$. Note that $\{f\}$ is also a nest-set of $H_0$ whereas $\{c\}$ is not since $\{a,b,d\}$ and $\{d,f\}$ are both in $\neight(\{c\},H_0)$ and clearly neither $\{a,b,d\} \subseteq \{d,f\}$ nor $\{a,b,d\}\supseteq \{d,f\}$ holds.

  In the second step of the elimination process we then consider $H_1 = H_0 - \{c,f\}$ and the nest-set $s_2=\{e,b\}$.
  It is again straightforward to verify that $\neight(s_2,H_1)=\{\{a,d\},\emptyset\}$ is linearly ordered by $\subseteq$.
  This is in fact the only nest-set of $H_1$.
  The third nest-set in the NEO, $s_3=\{a,d\}$ only becomes a nest-set
  after elimination of $s_2$: observe that $\neight(s_3,H_1)=\{\{e\},\{b\},\emptyset\}$ which is not linearly ordered by $\subseteq$.

  In the final step, $H_2 = H_1 - \{e,b\}$ only has two vertices
  left. The set of all vertices of a hypergraph is trivially a
  nest-set since $\neight(V(H),H)$ is always $\{\emptyset\}$. Thus, the set $V(H_2)=\{a,d\}$ is a nest-set of $H_2$. The hypergraph $H_0$ has no $1$-NEO (it has a $\beta$-cycle)
  and therefore $\nsw(H_0)=2$.
\end{example}

\begin{figure}[b]
  \centering
  \includegraphics[width=0.9\columnwidth]{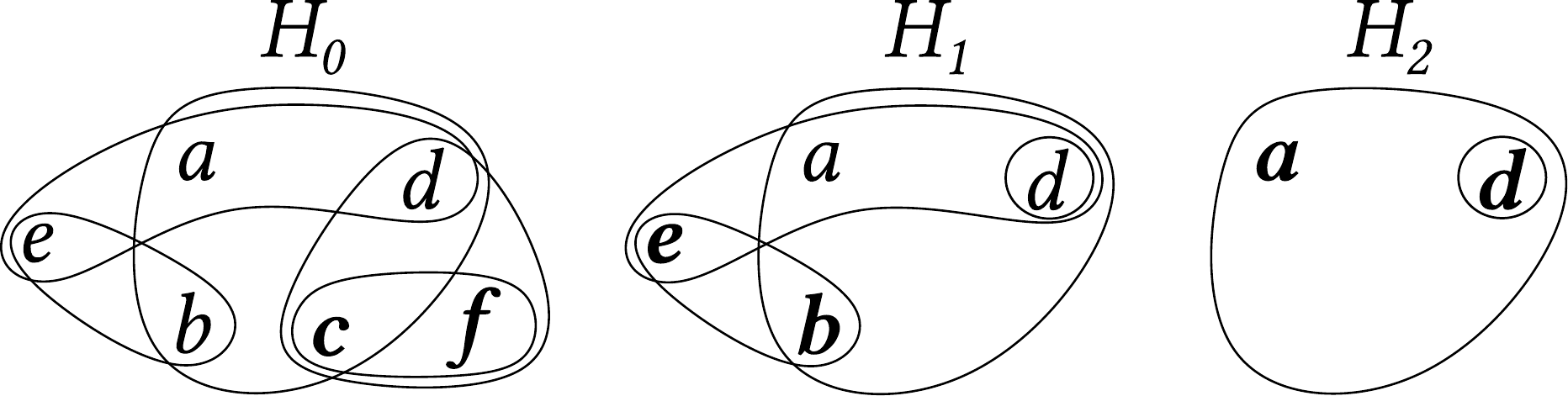}
    \Description{A visual representation of the hypergraphs described in Example~\ref{ex:nsw}}
  \caption{The nest-set elimination from Example~\ref{ex:nsw}}
  \label{fig:nswex}
\end{figure}

An important difference between $\alpha$- and $\beta$-acyclicity is
that only the latter is \emph{hereditary}, i.e., if hypergraph $H$ is
$\beta$-acyclic then so is every subhypergraph of $H$.  Nest-set
width, just like $\beta$-acyclicity and $\beta$-hypertree width, is indeed
also a hereditary property. In the following two simple but important lemmas, we first establish that NEOs remain valid
when vertices are removed from the hypergraph (and the NEO) and then
show that this also applies to removing edges.

Note that the construction in the following lemma, and Lemma~\ref{lem:nsw.hereditary} below, can technically create empty sets in the resulting NEOs. Formally speaking this is not allowed (recall that nest-sets are non-empty).
Whenever this occurs the implicit meaning is that all the empty sets are removed from the NEO.

\begin{lemma}
  \label{lem:nsw.induced.sub}
  Let $H$ be a hypergraph with $k$-NEO $\calO = (s_1, s_2, \dots, s_\ell)$ and let $r \subseteq V(H)$.
  Then the sequence $\calO' =(s_1 \setminus r, s_2 \setminus r, \dots, s_\ell \setminus r)$ is a $k$-NEO of $H-r$.
\end{lemma}
\begin{proof}

  We first show that for any nest-set $s$ let $r \subseteq V(H)$ we
  have that $s\setminus r$ is either the empty set or a nest-set of
  $H-r$.
  Suppose $s \setminus r$ is not empty and not a nest-set of $H-r$, then there are $e_1, e_2 \in \neigh(s\setminus r, H-r)$ that are not comparable by $\subseteq_{s \setminus r}$.
  It is easy to see that there exist $e'_1, e'_2 \in \neigh(s, H)$ such that $e_1 = e'_1 \setminus r$ and $e_2 = e'_2 \setminus r$. Since $s$ is a nest-set in $H$, w.l.o.g., $e'_1 \setminus s \subseteq e'_2 \setminus s$ and therefore also
  $$
  e_1 \setminus (s \setminus r) = e'_1 \setminus (s \cup r) \subseteq e'_2 \setminus (s \cup r) = e_2 \setminus (s \setminus r)
  $$
  and we arrive at a contradiction.

  It follows that $s_1 \setminus r$ is a $k$-nest-set of $H-r$.
  Since $\calO$ is a NEO, $s_2$ must be a nest-set of $H-s_1$.
  Now, to verify $\calO'$ we need to show that $s_2\setminus r$ is a $k$-nest-set of $H-r-s_1$. However, this is clearly
  the same hypergraph as $(H-s_1)-r$ and the above obsevation applies again. We can repeat this argument for all $s_i$ until $s_\ell$ and thus $\calO'$ is a $k$-NEO.
\end{proof}

\begin{lemma}
  \label{lem:nsw.hereditary}
  Let $H$ be a hypergraph with $k$-NEO $\calO = (s_1, \dots, s_\ell)$. Let $H'$ be a connected subhypergraph of $H$
  and $\Delta = V(H) \setminus V(H')$ the set of vertices no longer present in the subhypergraph.
  Then the sequence $(s_1 \setminus \Delta, s_2 \setminus \Delta, \dots s_\ell \setminus \Delta)$ is a $k$-NEO of $H'$.
\end{lemma}
\begin{proof}
  From the argument at the beginning of the proof of Lemma~\ref{lem:nsw.induced.sub} we know that $s \setminus \Delta$ is empty or a nest-set of $H-\Delta$. Therefore, $\neight(s \setminus \Delta, H-\Delta)$ has a linear order under $\subseteq$. Now, since $H'$ does not contain any vertices from $\Delta$ and is a subhypergraph of $H$ we have $E(H')=E(H'-\Delta) \subseteq E(H-\Delta)$ and thus $\neight(s \setminus \Delta, H') \subseteq \neight(s \setminus \Delta, H-\Delta)$. Therefore $\neight(s \setminus \Delta, H')$ can be linearly ordered by $\subseteq$ and thus $s \setminus \Delta$ is a nest-set.
  This observation can again be iterated along the NEO in the same fashion as in the proof of Lemma~\ref{lem:nsw.induced.sub} to prove the statement.
\end{proof}

\section{Nest-Set Width vs. $\beta$-Hypertree Width}
\label{sec:betahw}

A wide variety of hypergraph width measures have been studied in the
literature. To provide some context for the later algorithmic results,
we will first investigate how $\nsw$ relates to a number of prominent
width notions from the literature.
In particular, in this section we show that
$\nsw$ is a specialization of $\beta$-hypertree width and incomparable
to primal and incidence clique width and treewidth. The relationship
to $\beta$-hypertree width is of particular interest since bounded
\bhw also generalizes $\beta$-acyclicity. The section is structured
around proving the following theorem.

\begin{theorem}
  \label{thm:betahw}
  Bounded $nsw$ is a strictly less general property than bounded $\beta$-hw. In particular, the following two statements hold:
  \begin{enumerate}
  \item For every hypergraph $H$ we have $\beta$-hw$(H) \leq 3\, \nsw(H)+1$.
    \label{thm:betahw.1}
  \item There exists a class of hypergraphs with bounded $\beta$-hw and unbounded $nsw$.
    \label{thm:betahw.2}
  \end{enumerate}
\end{theorem}

We begin by establishing a useful technical lemma that will eventually lead us to the second statement of
Theorem~\ref{thm:betahw}. An important consequence of the following Lemma~\ref{lem:cycle} is 
 that the length (minus 1) of the longest $\beta$-cycle of $H$
is a lower bound of $nsw(H)$ since any vertex in a cycle has to be removed at some point in any NEO.

\begin{lemma}
  \label{lem:cycle}
  Let $C = (e_1, v_1,e_2,v_2 \dots, e_\ell,v_{\ell}, e_{\ell+1})$ be a $\beta$-cycle in a hypergraph $H$. For every nest-set $s$
  of $H$ we have that $|s \cap \{v_1,\dots,v_\ell\}|$ is either $0$ or at least $\ell-1$.
\end{lemma}
\begin{proof}[Proof of Lemma~\ref{lem:cycle}]
  Suppose the cardinality of $s \cap \{v_1, \dots, v_\ell\}$ is not
  $0$. That is, at least one vertex of $C$ is in $s$. Since we can rotate the indices of a cycle arbitrarily we assume, w.l.o.g., that $v_1 \in s$. Then,  $e_2$ and $e_\ell$ are both in $\neigh(s)$. Recall that a $\beta$-cycles has $\ell \geq 3$ and that $v_2$ can occur only in $e_1$ and $e_2$ and no other edges. Similarly, $v_\ell$ can occur exclusively in $e_{\ell-1}$ and $e_{\ell}$.
  We therefore see that $v_2 \not \in e_\ell$ and $v_\ell \not \in e_2$. Thus, $e_2$ and $e_\ell$ can only be comparable by $\subseteq_s$ if at least one of $v_2$ or $v_\ell$ is in $s$.
  
  Suppose, w.l.o.g.,  $v_2 \in s$, then we have $e_3$ and $e_\ell$ in $\neigh(s)$ and the same argument can be applied again, as long as the two edges are not adjacent in the cycle. We can then apply the argument exhaustively, until all edges of the cycle are in $\neigh(s)$ at which point it is clear that at least $\ell-1$ vertices are necessarily in $s$.
\end{proof}

Lemma~\ref{lem:cycle} further emphasizes the aforementioned
distinction between generalizing acyclicity in sense of tree-likeness
and our approach. Any cycle graph $C_n$ has hypertree width $2$
whereas the lemma shows us that $\nsw(C_n) \geq n-1$ since any
nest-set will contain at least one vertex of the cycle, so it must
contain at least $n-1$ of them. Furthermore, cycle graphs have
clique width at most 4~\cite{DBLP:journals/dam/CourcelleO00} and
treewidth at most 2. We therefore arrive at the following lemma.

\begin{lemma}
  \label{lem:cycleclass}
  There exists a class of hypergraphs that has bounded \bhw, treewidth, and clique width and unbounded $\nsw$.
\end{lemma}

The lemma establishes the second statement of Theorem~\ref{thm:betahw}.  We can derive some further results by combining Lemma~\ref{lem:cycleclass} with results from~\cite{DBLP:conf/icalp/GottlobP01}. There it was shown 
that there exist classes of $\beta$-acyclic hypergraphs that have unbounded clique width and treewidth. In combination with the previous lemma this demonstrates that bounded clique width and bounded treewidth are incomparable to bounded $\nsw$.
The results in~\cite{DBLP:conf/icalp/GottlobP01} also apply to incidence clique width and incidence treewidth and since the incidence graph of a cycle graph is also a cycle graph, so does Lemma~\ref{lem:cycleclass}. Thus, bounded $\nsw$ is also incomparable to bounded incidence clique width and bounded incidence treewidth.
The resulting hierarchy is summarized in Figure~\ref{fig:hierarchy} at the end of this section.

We move on to show that \bhw$(H) \leq 3\,nsw(H)+1$. We will give a procedure to construct a generalized hypertree decomposition of width $k$ from a $k$-NEO. Since $k$-NEOs are hereditary, every subhypergraph of $H$ will also have a generalized hypertree decomposition of width $k$. By a result of Adler, Grohe, and Gottlob in~\cite{DBLP:journals/ejc/AdlerGG07} we have that $\hw(H) \leq 3\,\ghw(H)+1$. From there we can then derive our bound of \bhw$(H) \leq 3\,k+1$. In particular, we make use of the observation that a nest-set is connected to the rest of the hypergraph only via its guard.
The necessary details of this observation are captured by the following two definitions and the key Lemma~\ref{lem:hinge} below.
The following construction is inspired by the hinge decompositions of Gyssens, Jeavons, and Cohen~\cite{DBLP:journals/ai/GyssensJC94}.

\begin{definition}[Exhaustive Subhypergraphs]
  Let $H$ be a hypergraph and $E' \subseteq E(H)$. Let $E^* := \{e \in E(H) \mid e \subseteq \bigcup E'\}$ be the edges covered by $E'$. Then we call the subhypergraph $H'$ with $E(H') = E(H) \setminus E^*$ the \emph{exhaustive $E'$-subhypergraph} of $H$.

  We use the term \emph{connected exhaustive $E'$-subhypergraphs} of $H$ to refer to the
  connected components of $H'$ (considering each component as an individual hypergraph).
\end{definition}

We use exhaustive subhypergraphs to express that, when we remove a set of edges $E'$ from $H$, then we also want to remove the edges $E^*$ that are covered by $\bigcup E'$. The following construction of a hypertree decomposition from a NEO will use sets of the form $\bigcup E'$ as its bags. This means that the respective bag also covers all edges in $E^*$. We are therefore interested in the components resulting from removing all of $E^*$ instead of just $E'$ from $H$.

In particular, we want to remove sets of edges $E'$ in such a way that the exhaustive $E'$-subhypergraphs are all connected to $E'$ via a single edge. This will allow us to bring together the decompositions of the subhypergraphs in a way that preserves all properties of hypertree decompositions. 

\begin{definition}[Exhaustive Hinges]
  Let $H$ be a hypergraph, $E' \subseteq E(H)$ and $C_1,\dots, C_n$ the connected exhaustive $E'$-subhypergraphs of $H$.
  For an $e \in E(H)$ we say that $E'$ is an \emph{exhaustive $e$-hinge} if for every $i \in [n]$
  we have that $V(C_i) \cap \bigcup E' \subseteq e$.
\end{definition}

\begin{lemma}
  \label{lem:hinge}
  Let $s$ be a $k$-nest-set of hypergraph $H$ and let $e_g$ be the guard of $s$. Then there
  exists an exhaustive $e_g$-hinge $E' \subseteq E(H)$ with the following properties:
  \begin{enumerate}

  \item $\bigcup \neigh(s, H) \subseteq \bigcup E'$
  \item $|E'| \leq k$
  \end{enumerate}
\end{lemma}
\begin{proof}
  Let $s$ and $e_g$ be as in the statement. Let $\lambda$ be a minimal
  edge cover of $s \setminus e_g$.  Observe that
  $|s \setminus e_g| < k$ as $e_g$ is incident to $s$ and therefore
  $|\lambda|<k$. We now claim that $E'=\lambda \cup \{e_g\}$ is the
  required hinge. Clearly we have $|E'|\leq k$. For the first property,
  recall that for every $e \in \neigh(s,H)$ we have
  $e \setminus s \subseteq e_g$ and thus also
  $e \subseteq e_g \cup s$. It is then easy to see from the definition
  of $E'$ that $e_g \cup s \subseteq \bigcup E'$ and the property
  follows.

  What is left to show is that that $E'$ is in fact an exhaustive $e_g$-hinge.
  Let $C$ be one of the connected exhaustive $E'$-subhypergraphs of $H$ and
  partition the set $V(C) \cap \bigcup E'$ in two parts:
  $I_1 := V(C) \cap s$ and
  $I_2 := V(C) \cap \left(\left(\bigcup E'\right) \setminus s\right)$.
  
  First we argue that $I_1 = \emptyset$. It was already established that $\bigcup \neigh(s, H) \subseteq \bigcup E'$,
  thus every edge incident to $s$ is removed in the exhaustive $E'$-subhypergraph. It is therefore impossible
  for a vertex of $s$ to be in $V(C)$.

  Second, observe that by construction every edge in $E'$ is incident to $s$ and
  by definition of the guard of $s$ we thus have
  $\left(\left(\bigcup E'\right) \setminus s\right) \subseteq e_g$. It
  follows immediately that $I_1 \cup I_2 \subseteq e_g$ and the
  statement holds. 
\end{proof}

Lemma~\ref{lem:hinge} is the key lemma for our construction procedure. It tells us that we can always find a \emph{small} exhaustive hinge $E'$ in a hypergraph $H$ if it has a $k$-NEO. By the first property from the lemma, the exhaustive $E'$ subhypergraph no longer contains the vertices $s$. From the connected exhaustive $E'$-subhypergraphs we can construct subhypergraphs of $H$ that connect to $E'$ via a single edge and have shorter $k$-NEOs than $H$. Since the subhypergraphs are connected to $E'$ via a single edge, it is straightforward to combine individual hypertree decompositions for every subhypergraph into a new decomposition for $H$.
This step can then be applied inductively on the length of the $k$-NEO to construct a hypertree decomposition of width $k$ for $H$.

\begin{lemma}
  \label{lem:hw.nsw}
  For any hypergraph $H$ it holds that $\ghw(H) \leq nsw(H)$.
\end{lemma}
\begin{proof}[Proof of Lemma~\ref{lem:hw.nsw}]
    We show by induction on $\ell \geq 1$ that if a hypergraph $H$ has a $k$-NEO of length $\ell$ then it has
  a generalized hypertree decomposition of width at most $k$.
  For the base case, $\ell=1$, the NEO
  consists of a single nest-set $s = V(H)$ with $|s|\leq
  k$. The base case then follows from the straightforward observation that
  $\ghw(H) \leq |V(H)|$.

  Suppose the statement holds for $\ell' <\ell$. We show that it also holds for every $k$-NEO of length
  $\ell$. Let $\calO=(s_1, \dots, s_\ell)$ be a $k$-NEO of
  $H$.  Let $e_g$ be the guard of $s_1$ and let $E'$ be the exhaustive $e_g$-hinge
  from Lemma~\ref{lem:hinge} and let $C'_1,\dots,C'_n$ be the connected exhaustive $E'$-subhypergraphs.
  Finally, for each $i \in [n]$, we add $e_g$ to $C'_i$ to obtain the hypergraph $C_i$.

  By Lemma~\ref{lem:nsw.hereditary} we see that for each $i \in[n]$,
  $C_i$ has a $k$-NEO $\calO_i = (s_{1,i}, \dots, s_{\ell,i})$ since
  it is a subhypergraph of $H$. Furthermore, according to Lemma~\ref{lem:nsw.hereditary},
  we can assume an $\calO_i$ such that $s_{1,i}\subseteq s_i$ and, since $e_g$ in $C_i$, also $s_{1,i} \neq \emptyset$.
  
  Therefore, $C_i-s_{1,i}$ has a $k$-NEO of length at most $\ell-1$ and
  we can apply the induction hypothesis to get a generalized hypertree decomposition
  $\left< T_i, (B_{u,i})_{u\in T_i}, (\lambda_{u,i})_{u\in T_i}\right>$ with $\ghw \leq k$ of $C_i-s_{1,i}$. Observe that the hypergraph
  has an edge $e_g \setminus s_{1,i}$ which has to be covered
  completely by some node $u_{g,i}$ in $T_i$.

  Let $u$ be a fresh node with $B_u = \bigcup E'$ and $\lambda_u = E'$.
  For each $i \in [n]$ we now change the root of $T_i$ to be $u_{g,i}$ and attach the tree as a child of $u$.
  A cover $\lambda_{w,i}$ of a node $w$ in $T_i$ can contain an edge $e'$ that are not in $H$ because the vertices $s_{1,i}$ are removed. As no new edges are ever added, the only possibility for $e'$ to not be in $E(H)$ is that $e' = e \setminus s_{1,i}$ for some $e \in E(H)$. We can therefore replace any such $e'$ by and edge $e\in E(H)$ in $\lambda_{w,i}$ in such a way that $B_{w,i}$ remains covered by $\lambda_{w,i}$. 

 We claim that this newly built decomposition is a generalized hypertree decomposition of $H$ with $\ghw \leq k$.
  It is not difficult to verify that this new structure indeed
  satisfies all proprieties of a generalized hypertree decomposition.

  \paragraph{Connectivity} Each subtree below the root already satisfies
    connectivity. The tree structure and the bags in the subtree remains unchanged. Furthermore, by construction of the hypergraphs $C_i$, the sets $B(T_i)$ of vertices occuring in bags of the tree $T_i$ are pairwise
    disjoint except for the vertices in $e_g$. Since $e_g$ is
    fully in $B_u$ the only issue for connectivity can arise if there
    is a vertex in $B_u \cap B(T_i)$ but not in $B_{u_{g,i}}$. We argue
    that this is impossible.
    
    Since $E'$ is an exhaustive $e_g$-hinge and $e_g$ was added back into each component
    it is easy to see that

    $$B_u \cap B(T_i) = \bigcup E' \cap V(C_i - s_{1,i}) \subseteq e_g \setminus s_{1,i}$$
    
    The rightmost term is exactly the edge that informed our choice of $u_{g,i}$, i.e., we have $B_u \cap B(T_i) = e_g \setminus s_{1,i} \subseteq B_{u_{g,i}}$ by construction.

    \paragraph{Every edge of $H$ is covered} For every edge
    $e \in E(H)$ we consider two cases. Either $e \in \neigh(s_1,H)$
    or not.  In the first case, by Lemma~\ref{lem:hinge} we have
    $e \subseteq \bigcup E'$ and therefore it is covered in the root
    node $u$. In the second case, $e \not \in \neigh(s_1,H)$, $e$ will
    occur unchanged in one of the hypergraphs $C_i - s_{1,i}$ since
    the removal of $s_{1,i}$ does not affect it (recall
    $s_{1,i}\subseteq s_i$).
    Since the tree decomposition of
    $C_i-s_{1,i}$ remain the same, except for changing which node is the root, $e$ must be
    covered in the respective subtree corresponding to component
    $C_i - s_{i,1}$.

\end{proof}

\begin{proof}[Proof of Theorem~\ref{thm:betahw} (\ref{thm:betahw.1})]
  By Lemma~\ref{lem:hw.nsw} we have that $\ghw(H) \leq \nsw(H)$. As mentioned above, we always have $\hw(H) \leq 3\,\ghw(H)+1$ for every hypergraph and therefore also $\hw(H)\leq 3\,\nsw(H)+1$. In combination with Lemma~\ref{lem:nsw.hereditary}
  we see that for every subhypergraph $H'$ of $H$ we have $hw(H') \leq 3\,nsw(H')+1 \leq 3\,nsw(H)+1$.
\end{proof}

The results of this section are summarized in Figure~\ref{fig:hierarchy}. The diagram extends the hierarchy given in~\cite{DBLP:conf/icalp/GottlobP01} by bounded $\nsw$.

\begin{figure}[t]
  \centering
  \includegraphics[width=0.8\columnwidth]{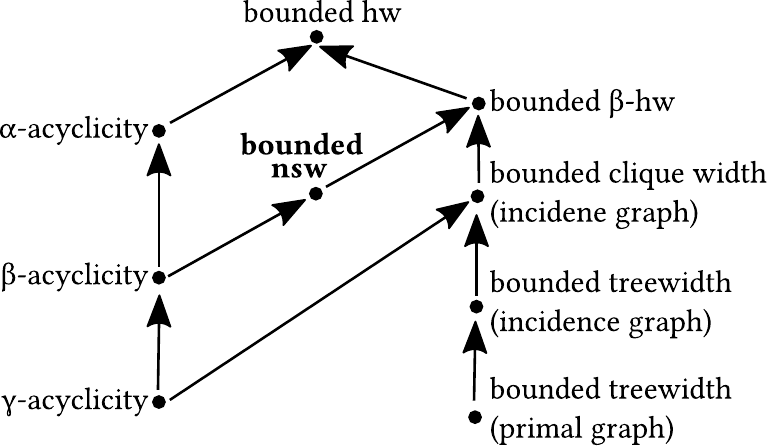}
  \Description{
    A hierarchy diagram for structural parameters summarising the results of this section. Importantly, nest-set width is strictly more general than $\beta$-acyclicity and less general than \bhw. At the same time, bounded \nsw is incomparable to bounded incidence clique-width.
  }
  \caption{Expressive power of various hypergraph properties from~\cite{DBLP:conf/icalp/GottlobP01}, extended by bounded $\nsw$. (Arcs are directed from less general to more general. Properties with no directed connection are incomparable.)}
  \label{fig:hierarchy}
\end{figure}

\section{The Complexity of Checking Nest-Set Width}
\label{sec:complexity}

For the existing generalizations of $\beta$-acyclicity -- \bhw and
$pw$ -- it is not known whether one can decide in polynomial time if
a structure has width $\leq k$, even when $k$ is a constant.  This
then also means that no efficient algorithm is known to compute the respective witnessing structures. In these situations, tractability results are
inherently limited. One must either assume that the witnesses are given
as an input or that a tractable algorithm does not use the witness at all.  In
comparison, deciding treewidth $\leq k$ is fixed-parameter tractable
when parameterized by $k$~\cite{DBLP:conf/cocoon/BodlaenderF96} and
checking hypertree width is tractable when $k$ is
constant~\cite{DBLP:journals/jcss/GottlobLS02}.

When $k$ is considered constant, it is straightforward to find a $k$-NEO in polynomial time, if one exists.
We can simply check for all combinations of up to $k$
vertices whether they represent a nest-set. If so, eliminate the nest-set and repeat from the beginning on the new hypergraph until it becomes empty. By Lemma~\ref{lem:nsw.induced.sub}, this greedy approach of always using the first found
$k$-nest-set will result in a sound and complete procedure.

However, we can improve on this straightforward case by analyzing the following decision problem where $k$ is part of the input.
\begin{problem}[framed]{\nswprob}
   Instance: & A hypergraph $H$, integer $k$ \\
   Question: & $nsw(H) \leq k$?
 \end{problem}
 We first observe that \nswprob is \np-complete in
 Section~\ref{sec:hardness}. In more positive news, we are able to
 show that \nswprob is fixed-parameter tractable when parameterized by
 $k$ in Section~\ref{sec:fpt}.  Importantly, the fixed-parameter
 algorithm explicitly constructs a $k$-NEO as a witness, if one
 exists, and can therefore serve as a basis for the algorithmic
 results in the following sections.

\subsection{\np-Hardness}
\label{sec:hardness}

In this section the following complexity result is shown.
\begin{theorem}
  \label{thm:npc}
  \nswprob is \textsc{NP}-complete.
\end{theorem}
\np-hardness is demonstrated by reduction from \vcov, a classical \np-complete
problem~\cite{DBLP:conf/coco/Karp72}. 
In the \vcov problem we have as input a
graph $G=(V,E)$ and an integer $k\geq 1$. The problem is to
decide whether there exists a set $\alpha \subseteq V$ with
$|\alpha|\leq k$ such that every edge of $G$ is incident to at least
on vertex in $\alpha$. Such a set $\alpha$ is called a \emph{vertex cover} of $G$. 
To simplify the following argument we make two additional assumptions on
the instances of \vcov.  We assume that the input graph has at least 2
edges and that $k$ is strictly less than the number of edges in
$G$. If either assumption is violated the problem is trivial.

We first prove that it is \textsc{NP}-complete to decide whether a hypergraph has a $k$-nest-set.
The hardness of \nswprob then follows from the argument below, that shows that the hypergraph $H$ constructed in the reduction has a $(km+k)$-nest-set if and only if it has nest-set width at most $km+k$.

\begin{lemma}
  Deciding whether a hypergraph $H$ has a $k$-nest-set is \textsc{NP}-complete.
\end{lemma}
\begin{proof}
  Membership is straightforward. Guess up to $k$ vertices $s$ and
  verify the orderability of $\neight(s,H)$. Hardness is by many-one
  reduction from \vcov.  Hence, let $G,k$ be an instance of \vcov and
  let $n$ and $m$ refer to the number of vertices and edges in $G$,
  respectively. In the following we construct a hypergraph $H$ such that
  $H$ has a $(km+k)$-nest-set if and only if $G$ has a vertex
  cover of size at most $k$.

  Let $\{e_1, \dots, e_m\}$ be the edges of $G$ and let $\{v_1,\dots, v_n\}$ be the vertices of $G$.
  Our $H$ will have as vertices $V(H) = \bigcup \{v_j^{(i)} \mid j \in [n], i \in [m+1]\}$ a copy of every vertex
  $v_j$ associated to edge $e_i$. We will refer to the superscript $(i)$ also as the $i$th \emph{level} of $H$.
  We will write $V^{\leq i}$ for $\{v_j^\ell \in V(H) \mid \ell \leq i\}$, i.e., all the vertices at level $i$ or lower.

  For each edge $e_i=\{a,b\}$ of $G$, we create two edges $f_{i,1}$ and $f_{i,2}$ in $H$ as follows.
  \[
    \begin{aligned}
      & f_{i,1} = V^{\leq i} \setminus \{a^{(i)}\} \qquad  & f_{i,2} = V^{\leq i} \setminus \{b^{(i)}\} %
    \end{aligned}
  \]
  Furthermore, we also add two edges  $f_{m+1,1} = V^{\leq m+1}\setminus \{a^{(m+1)}\}$ and $f_{m+1,2} = V^{\leq m+1}\setminus \{b^{(m+1)}\}$ at the final level for $e_1 = \{a,b\}$.  
  Intuitively, these $f$ edges at level $i$ represents the choice between $a$ and $b$ for edge $e_i$, as one
  needs to be deleted for the two edges to be comparable by
  $\subseteq$. We will therefore refer to them as the \emph{choice edges}.
  Encoding the choice for $e_1$ twice, at levels $1$ and $m+1$, is done for technical reasons that will become apparent later.
  $H$ also contains the complete graph $K^{(j)}$ over the vertices $\{v_j^{(i)}\mid i\in [m+1]\}$ for every vertex $v_j$ of $G$.
  Intuitively, they link the choices at every level to each other and we therefore refer to them as the \emph{linking cliques}. Thus we have
  \[
    E(H) = \{f_{i,1},f_{i,2} \mid i \in [m+1] \} \cup \bigcup_{j \in [n]} E(K^{(j)})
  \]

  We first show that if $\alpha$ is a vertex cover of $G$ and $\ell \leq k$, then
  $s_\alpha = \bigcup_{v_j \in \alpha} \{ v_j^{(i)}\mid i\in[m+1]\}$ is a $(km+k)$-nest-set of
  $H$. Note that $|s_\alpha|\leq k(m+1)$ follows immediately from the construction.
  \begin{claim}
    \label{claim:npa}
    For each $i \in [m+1]$ we have
    $f_{i,1} \subseteq_{s_\alpha} f_{i,2}$ or vice versa.
  \end{claim}
  \begin{claimproof}
    Let $e_i = \{a,b\}$ and, w.l.o.g., assume $a \in \alpha$ and thus
    also $a^{(i)} \in s_\alpha$.  Clearly,
    $f_{i,2} \setminus\{a^{(i)}\} = V^{\leq i}\setminus
    \{a^{(i)},b^{(i)}\} \subseteq f_{i,1}$. Removing further vertices
    from both can not change the order anymore and thus $f_{i,2}\setminus s_\alpha \subseteq f_{i,1} \setminus s_\alpha$.
    If $b \in \alpha$ the analogous argument yields the opposite order. The same argument also applies for $f_{m+1,1}$ and $f_{m+1,2}$ and $e_1$.
  \end{claimproof}
 
  By construction we always have $f_{i,1},f_{i,2} \supseteq f_{h,1},f_{h,2}$ for $h < i$. Thus, in combination with Claim~\ref{claim:npa} we have that all choice edges are linearly orderable
  by $\subseteq_{s_\alpha}$. The only other edges in $\neigh(s_\alpha)$ are those of the linking cliques $K^{(j)}$ where $v_j \in \alpha$.
  Clearly, all edges of the clique become empty, as $V(K^{(j)})\subseteq s_\alpha$,  and thus $\neight(s_\alpha)$ is linearly orderable by $\subseteq$.

  For the other direction, suppose $s$ is a $(km+k)$-nest-set of $H$.
  We now define $\alpha$ as containing exactly those vertices $v_j$ such that $v_j^{(i)} \in s$ for at least $m$ distinct $i$. Note that because of the linking
  cliques and Lemma~\ref{lem:cycle}, a vertex $v_j^{(i)}$ occurs
  either for 0 or at least $m$ distinct $i$ in $s$. Using the assumptions that $2 \leq m$ and  $1 \leq k < m$ from above it is straightforward to verify that $\frac{km+k}{m}< k+1$ and therefore also $|\alpha|\leq k$.

  What is left is to show that every edge $e_i$ in $G$ is incident to
  a vertex in $\alpha$.  For any
  $e_i = \{a,b\}$ the choice edges $f_{i,1},f_{i,2}$ are only comparable by $\subseteq_s$ if
  either $a^{(i)} \in s$ or $b^{(i)} \in s$ (or both). Then,
  because of the linking cliques, $a$ or $b$ (or both) will be in
  $\alpha$.

  It follows that if $f_{i,1},f_{i,2} \in \neigh(s,H)$,
  then $e_i$ will be covered by $\alpha$. Then, since $s$ is not empty, there is some $v_j$ such
  that $v_j^{(i)}$ occurs for $m$ distinct $i$ in $s$. In particular, then either $v_j^{(1)}$ or $v_j^{(2)}$
  are in $s$ and thus for every $2 < i \leq m+1$ we have $f_{i,1},f_{i,2} \in\neigh(s)$. Thus, for every edge in $G$ there is a pair of choice edges in $\neigh(s)$. By the argument above every edge of $G$ is therefore incident to a vertex in $\alpha$. 
\end{proof}

We now build on this reduction to prove Theorem~\ref{thm:npc}.
Suppose the same situation as in the above proof, i.e., a vertex cover
instance $G, k$ and the hypergraph $H$ from the reduction above. If
$G$ has a vertex cover $\alpha$, then we can eliminate all the
vertices $s_\alpha$ in $H$ that encode the graph vertices from
$\alpha$. By the argument above we have that $E(H) = \neigh(s_\alpha)$
is linearly ordered by $\subseteq_{s_\alpha}$. It is not difficult to see that when all edges of a hypergraph are linearly ordered by $\subseteq$ then that hypergraph is $\beta$-acyclic: if a vertex $v$ is included in every edge then $\{v\}$ is a nest-set. 
Thus, $H-s_\alpha$ has a $1$-NEO $\calO'$ and thus prepending $s_\alpha$ to $\calO'$ gives us a $(km+k)$-NEO of $H$.

On the other hand, suppose $H$ has a $(km+k)$-NEO
$\calO = (s_1,\dots)$.  As $s_1$ is a $(km+k)$-nest-set of $H$, the
arguments from the proof apply and we have that $G$ has a vertex cover
of size at most $k$. This now also completes the proof of Theorem~\ref{thm:npc}.

On a final note, one may notice similarities between finding nest-sets and an important work by Yannakakis~\cite{DBLP:journals/siamcomp/Yannakakis81a} on vertex-deletion problems in bipartite graphs. Yannakakis gives a complexity characterization for vertex-deletion problems on bipartite graphs that extends to hypergraphs via their incidence graph. Furthermore, the specific problem of finding a vertex-deletion such that the edges of the hypergraph become linearly ordered by $\subseteq$ is stated to be polynomial. While this strongly resembles the nest-set problem, the results of Yannakakis are not applicable here since we are not interested in a global property of the hypergraph but only in the orderability of the edges that are incident to the deleted vertices.

\subsection{Fixed-parameter Tractability}
\label{sec:fpt}

Recall that every nest-set $s$ has a maximal edge with respect to
$\subseteq_s$; the guard of $s$. The main idea behind the algorithm
presented in this section is to always fix an edge $e_g$ and check if
there exists a nest-set that specifically has $e_g$ as its guard.
This will allow us to incrementally build a nest-set $s$ relative to
the guard $e_g$. We first demonstrate this principle in the following
example.

\begin{example}
  \label{ex:fpt}
  We consider a hypergraph $H$ with three edges $e_1=\{a,b,c,d\}$, $e_2=\{a,b,c,g\}$, and $e_3=\{c,d,g,f\}$. We want to find a nest-set with guard $e_1$. The hypergraph with $e_1$ highlighted is shown in Figure~\ref{ex:fpt}. To start, if $s$ is a nest-set with guard $e_1$, then at least one vertex of $e_1$ must be in $s$. For this example let $a \in s$.

  Since $a \in s$ we also have that $e_2 \in \neigh(s)$. For $s$ to be a nest-set with guard $e_1$ it must then hold that $e_2 \setminus s \subseteq e_1 \setminus s$. Since $g$ is in $e_2$ but not in $e_1$ we can deduce that also $g \in s$. More generally, any vertex that occurs in an edge from $\neigh(s)$ but not in $e_1$ must be part of the nest-set $s$. Now, since $g \in s$ it follows that $e_3 \in \neigh(s)$ and therefore, by the previous observation, also $f \in s$.

  At this point we have deduced that if $a$ is in $s$, then so are $g$ and $f$. We now have the situation that for every edge $e \in \neigh(s)$ we have $e \setminus \{a,g,f\} \subseteq e_1 \setminus \{a,g,f\}$. However, as illustrated in Figure~\ref{ex:fpt},
  $\neight(\{a,g,f\})$ is not linearly ordered by $\subseteq$. A nest-set must therefore contain further vertices. In this case it is easy to see that either removing $b$ or $d$ is enough. In conclusion we have shown that if $a \in s$, then there are two $4$-nest-sets $\{a,b,e,f\}$ and $\{a,d,e,f\}$ that have guard $e_1$.
\end{example}

\begin{figure}[t]
  \centering
  \includegraphics[width=0.65\columnwidth]{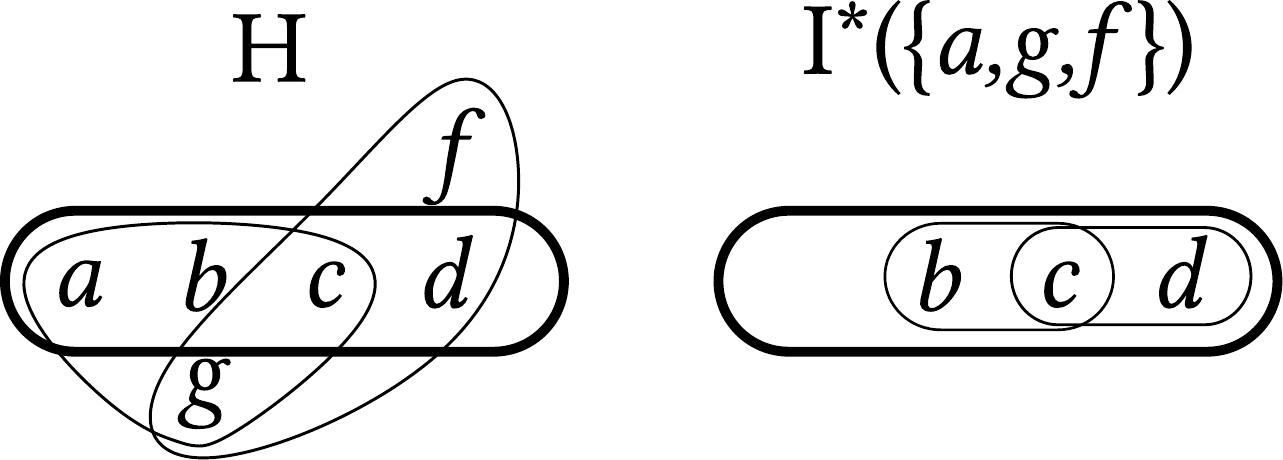}
  \Description{
    A visual representation of hypergraph $H$ and the sets of $\neight(\{a,g,f\})$ from Example~\ref{ex:fpt}.
  }
  \caption{Illustration of Example~\ref{ex:fpt}}
  \label{fig:fpt}
    \Description{A visual representation of the hypergraphs described in Example~\ref{ex:fpt}}
\end{figure}

What makes the problem difficult is that there can be many possible ways of making edges linearly ordered by vertex deletion. In Example~\ref{ex:fpt} both choices, removing either $b$ or $d$, lead to a $4$-nest-set. However, suppose there were an additional edge $e_4=\{b,x\}$. Then, choosing $b$ would also imply $e_4 \in \neigh(s)$ and $x \in s$. Choosing $d$ would lead to a smaller nest-set.

In general, this type of complication can occur repeatedly and it is
therefore necessary to continue this expansion procedure for all
possible (minimal) ways of ordering the known incident edges of $s$.
We will therefore first establish an upper bound on these possible expansions.

Intuitively, when we have edges $\{a,b\}$ and
$\{b, c\}$, the only way they become comparable by $\subseteq$ is if
either $a$ or $c$ is removed. The existence of a linear order over all
the edges thus requires resolving all such conflicts. By encoding
these conflicts in a kind of conflict graph we can see that the
problem is equivalent to finding a vertex cover in the conflict graph.

\begin{definition}[$\subseteq$-conflict graph]
  Let $H$ be a hypergraph, we define the \emph{$\subseteq$-conflict graph} of $H$ as the graph obtained by the following construction (with $V(G) = \bigcup E(G)$):
  For every two distinct edges $e_1, e_2 \in E(H)$, if $v \in e_1 \setminus e_2$ and $u \in e_2 \setminus e_1$, then add an
  edge $\{v,u\}$ to $G$. We say that $u$ and $v$ have a \emph{$\subseteq$-conflict} in $H$.
\end{definition}

\begin{lemma}
  \label{lem:enum.conflicts}
  Let $H$ be a hypergraph and let $s \subseteq V(H)$.
  Then $E(H-s)$ is linearly ordered by $\subseteq$ if and only if $s$ is a vertex cover of the $\subseteq$-conflict graph of $H$.
\end{lemma}
\begin{proof}
  Let $G$ be the $\subseteq$-conflict graph of $H$.
  We first show the implication from right to left. Let $s$ be a
  vertex cover for $G$ and suppose that $E(H-s)$ is not
  linearly ordered by $\subseteq$.  Hence,
  there are two edges $e_1, e_2 \in E(H-s)$ that are incomparable,
  i.e., there exist vertices $v\in e_1 \setminus e_2$ and
  $u \in e_2 \setminus e_1$ and neither $v$ nor $u$ is in the vertex
  cover $s$.  A conflict can not be introduced by removing the
  vertices of $s$ and therefore it was already present in $H$. Therefore,
  there must be an edge $\{v,u\}$ in $G$ that is not covered by
  $s$, contradicting that $s$ is a vertex cover.

  For the other direction let $s \subseteq V(H)$  such that $E(H-s)$ is linearly ordered by
  $\subseteq$. Then for every $\subseteq$-conflict, i.e., every pair of
  vertices $u,v$ where there are $e_1, e_2 \in E(H)$ with
  $v \in e_1 \setminus e_2$ and $u \in e_2 \setminus e_1$, at least
  one of $u,v$ must be in $s$.  All edges of $G$ are exactly between such pairs of vertices,
  hence $s$ contains at least one vertex of each edge in $G$. Therefore $s$ is also a vertex cover of $G$.
\end{proof}

This correspondence allows us to make use of the following
classical result by Fernau~\cite{DBLP:conf/cocoon/Fernau02} on the enumeration of all \emph{minimal vertex covers}. A vertex cover is called a minimal vertex cover if
none of its subsets is a vertex cover.

\begin{proposition}[\cite{DBLP:conf/cocoon/Fernau02}]
  \label{prop:vcov.enum}
  Let $G$ be a graph with $n$ vertices. There exist at most $2^k$
  minimal vertex covers with size $\leq k$ and they can be
  fully enumerated in $O(2^k k^2 + kn)$ time.
\end{proposition}

In combination with Lemma~\ref{lem:enum.conflicts} we therefore also have an upper bound on computing all minimal vertex deletions that resolve all $\subseteq$-conflicts. With this we are now ready to state Algorithm~\ref{alg:nsw.expand} which implements the intuition described at the beginning of this section. The algorithm is given a hypergraph and an edge $e_g$ to use as guard and tries to find a $k$-nest-set with guard $e_g$ by exhaustively following the steps described in Example~\ref{ex:fpt}.
We are able to show that this indeed leads a correct procedure for finding $k$-nest-sets with a specific guard. 

\newcommand{\funcname}{NestExpand}

\begin{algorithm}[t]
  \SetKwInput{Output}{output}
\SetKwInput{Input}{input}

\Input{Hypergraph $H$, edge $e_g$, and an integer $k \geq 1$.}
\Output{``Accept'', if there exists a nest-set $s$ with guard $e$ and $|s|\leq k$ \linebreak
        ``Reject'', otherwise.}

  \SetKw{Halt}{Halt}\SetKw{Reject}{Reject}\SetKw{Accept}{Accept}
  \SetKwProg{Fn}{Function}{}{}
  \SetKwFunction{greed}{\funcname}

  \Fn{\greed($s:$ set of vertices)}{
    \If{$|s|>k$}{
      \KwRet \Reject
    }
    $\Delta \leftarrow \bigcup \neigh(s, H) \setminus (s \cup e_g)$\;
    \If{$\Delta \neq \emptyset$}{
      \KwRet \greed($s \cup \Delta$)\; %
    }
    $H_g \leftarrow $ The hypergraph $\neight(s, H)$\; %
    \If{$H_g$ has no $\subseteq$-conflicts}{
      \KwRet \Accept\; %
    }
    $\mathbf{A} \leftarrow $ all minimal vertex covers of the $\subseteq$-conflict graph of $H_g$ with size at most $k-|s|$\;
    \ForEach{$\alpha \in \mathbf{A}$}{
      \If{\greed($s \cup  \alpha$) accepts}{
        \KwRet \Accept\; %
      }
    }
    \KwRet \Reject\;
  }

  \Begin(\tcc*[f]{\bfseries Main}){
    \ForEach{$v \in e_g$}{
      \If{\greed($\{v\}$)}{
        \KwRet \Accept\;
      }
    }
    \KwRet \Reject\;
  }
  
\caption{Find nest-set with fixed guard.}
\label{alg:nsw.expand}
\end{algorithm}                                                                                          

\begin{lemma}
  \label{lem:fpt.correct}
  Algorithm~\ref{alg:nsw.expand} is sound and complete.
\end{lemma}
\begin{proof}
  The algorithm has one base-case for acceptance, when
  $\Delta = \emptyset$ and $\neight(s,H)$ is linearly orderable by $\subseteq$. Clearly then $s$ is a nest-set.
  As $\Delta=\emptyset$, for every (sub)edge $e \in \neight(s,H)$ we have $e \subseteq e_g$, i.e., $e_g$ is a guard of $s$.
  From the check at the beginning of the \textsf{\funcname}
  we have $|s|\leq k$. Hence, if the algorithm accepts then the current $s$ is a $k$-nest-set with guard $e_g$.

  To establish completeness we show that if a $k$-nest-set $s$ with guard $e_g$ exists, then the algorithm will accept.
  In particular we claim that for every call \textsf{\funcname}$(s')$, if there exists a $k$-nest-set $s$ with guard $e_g$ and $s' \subseteq s$,
  then either $s'$ is a $k$-nest set or $s \cup X \subseteq s$ where $s \cup X$ is the parameter of one of the calls made directly by \textsf{\funcname}$(s')$.

  We distinguish two cases. First suppose there are edges $e \in \neigh(s')$ such that $e\setminus s' \not \subseteq e_g \setminus s'$. Since $e_g$ is the guard of $s$, and $\neigh(s')\subseteq \neigh(s)$, every element of  $e \setminus e_g$ must necessarily also be in $s$. This corresponds directly to the set $\Delta$ in the algorithm. Hence, $s' \cup \Delta \subseteq s$ when $\Delta \not \emptyset$, which is clearly the only parameter of a child call.

  In the other case, there are no such edges. The claim then states that either $s'$ is a $k$-nest-set and the algorithm accepts, or that $s' \cup \alpha  \subseteq s$ for some $\alpha \in \mathbf{A}$.
  Now suppose the claim is false, i.e., there exists a $k$-nest-set $s \supseteq s'$ with guard $e_g$ such that $s'$ is not a $k$-nest-set and $\forall \alpha \in \mathbf{A}. s' \cup \alpha \not\subseteq s$. By Lemma~\ref{lem:enum.conflicts}, $\mathbf{A}$ contains all possible minimum deletions with size at most $k-|s'|$ that make $\neight(s')$ linearly orderable by $\subseteq$. It follows that some $\subseteq$-conflict from $H_g$ must remain in $\neight(s)$ as otherwise some $\alpha \in \mathbf{A}$ would be a subset of $s$. This contradicts the fact that that assumption $s$ is a nest-set and thereby proves the claim.

  With the claim established, completeness then follows from the fact if $e_g$ is a guard of nest-set $s$, then $e_g \cap s \neq \emptyset$. Hence, if there exists a $k$-nest-set $s$ with guard $e_g$ it must contain at least one vertex $v$ of $e_g$. Inductive application of the claim then proves that a $k$-nest-set will be found by the algorithm (and accepted) when starting from
  \textsf{\funcname}$(\{v\})$.
\end{proof}

For the sake of simplicity, Algorithm~\ref{alg:nsw.expand} is stated as a decision procedure. Even so, it is easy to see that a $k$-nest-set with the appropriate guard has been constructed at any accepting state. It is then straightforward to use Algorithm~\ref{alg:nsw.expand} to decide in fixed-parameter polynomial time if a hypergraph has any $k$-nest-set, and if so output one. In the following we use $\norm{H} = |V(H)|+|E(H)|$ for the \emph{size} of hypergraph $H$.

\begin{theorem}
  \label{thm:fpt.nestset}
  There exists a $2^{O(k^2)}poly(\norm{H})$ time algorithm that takes
  as input hypergraph $H$ and integer $k\geq 1$ and returns a
  $k$-nest-set $s$ of $H$ if one exists, or rejects
  otherwise.
\end{theorem}
\begin{proof}
  We simply call Algorithm~\ref{alg:nsw.expand} once for each edge of
  $H$ as the guard. Since every nest-set has a guard
  Lemma~\ref{lem:fpt.correct} implies that this will find an
  appropriate nest-set if one exists. If all calls reject, then there
  can be no nest-set with at most $k$ elements as it is not guarded by
  any edge of $H$.

  \sloppy
  What is left to show is that Algorithm~\ref{alg:nsw.expand}
  terminates in $2^{O(k^2)}poly(k \norm{H})$ time. Calling the procedure $|E(H)|$
  times clearly preserves this bound. 
  First, observe that every recursive call of \textsf{\funcname} increases the cardinality of $s$ by at least one.
  The call tree of the recursion therefore has maximum depth $k$. Furthermore, by Proposition~\ref{prop:vcov.enum} every node in the call tree has at most $2^k$ children if $\Delta = \emptyset$, or exactly one when $\Delta\neq \emptyset$. Hence, at most $2^{(k^2)}|e_g|$ calls to \textsf{\funcname} are made in one execution of Algorithm~\ref{alg:nsw.expand}.

  In each call, the computation of $\Delta$ and $H_g$ as well as all
  the checks are feasible in $O(poly(\norm{H}))$ time. Since $\norm{H_g} \leq \norm{H}$, the set $\mathbf{A}$
  can be computed in $O(2^kk^2 + poly(k\norm{H}))$ time according to
  Proposition~\ref{prop:vcov.enum}. Recall that we assume $k \leq V(H)$ since the problem is trivial otherwise.
  The overall execution time of Algorithm~\ref{alg:nsw.expand}
  is therefore in $2^{O(k^2)}poly(\norm{H})$.
\end{proof}

Once we can find individual $k$-nest-sets, finding $k$-NEOs becomes simple. Recall from Lemma~\ref{lem:nsw.induced.sub}
that vertex removal preserves $k$-NEOs. Thus straightforward greedy removal of $k$-nest-sets is a sound and complete algorithm for finding $k$-NEOs. Since at most $|V(H)|$ nest-set removals are required to reach the empty hypergraph, using the procedure from Theorem~\ref{thm:fpt.nestset} to find the $k$-nest-sets yields a $2^{O(k^2)}poly(\norm{H})$ time algorithm for \nswprob.

\begin{corollary}
  \label{cor:fpt}
  \nswprob parameterized by $k$ is fixed-parameter tractable. 
\end{corollary}

\section{Nest-Set Width \& Conjunctive Queries with Negation}
\label{sec:scq}
\label{sec:safeneg}

We move on to prove our main algorithmic result. Recall that a query $q$ has an associated hypergraph $H(q)$.
We define the nest-set width of the query $q$ as $\nsw(q) = \nsw(H(q))$. We say that a class $\mathcal{Q}$ of \cqnegprob instances has \emph{bounded $\nsw$} if there exists a constant $c$, such that every query $q$ in $\mathcal{Q}$ has $\nsw(q) \leq c$.

\begin{theorem}
  \label{thm:safe.scq}
  For every class $\mathcal{Q}$ of \cqnegprob instances with bounded $\nsw$,
  \cqnegprob is decidable in polynomial time.
\end{theorem}

As usual, the result can be extended to unions of conjunctive
queries with negation (U\cqneg) when the $\nsw$ of a U\cqneg is
defined to be the maximum $\nsw$ of its \cqneg parts.

While the complexity of
CQs without negation has been extensively studied and is well
understood, few results extend to the case where negation is
permitted. When there are only positive literals, then the satisfying
assignments for each literal are explicitly present in the
database. Finding a solution for the whole query thus becomes a question of finding a
consistent combination of these explicitly listed partial
assignments. However, with negative literals it is possible to implicitly
express a large number of satisfying assignments. Recovering an
explicit list of satisfying assignments for a negative literal may
require exponential (in the arity) time and space because there can be up to $|Dom|^a$ such assignments, where $a$ is the arity of the literal.

This additional expressiveness of negative literals has
important implications for the study of structural parameters. While
evaluation of CQs is \np-complete with and without negation, permitting negation
allows for expressing problems as queries with a \emph{simpler} hypergraph structure.
Such a change in expressiveness relative to structural complexity must also
 be reflected in structural parameters that capture tractable classes
 of the problem.

 The following example and theorem illustrate this change in expressiveness relative to hypergraph structure.
 Consider \sat for propositional formulas in conjunctive normal form (CNF).
 Recall, that for a formula $F$ in CNF, the corresponding hypergraph $H(F)$ has as its vertices the variables of the formula and every edge is the set of variables of some clause in the formula.
 A clause $C=l_1 \lor \cdots \lor l_q$ has $2^q-1$ satisfying assignments to the variables of the clause. Thus, a corresponding positive literal in a CQ, that contains all the satisfying assignments, will be of exponential size (unless the size of clauses is considered bounded). On the other hand, there is a single assignment to $vars(C)$ that does not satisfy $C$. It is therefore possible to compactly encode \sat by having a negative literal for each clause that excludes the respective non-satisfying assignment. Since this reduction preserves the hypergraph structure of the \sat formula it follows that structural restrictions can only describe a tractable fragment of \cqnegprob if they also make \sat tractable. For example, \sat is \np-hard when restricted to $\alpha$-acyclic formulas~\cite{DBLP:journals/tcs/OrdyniakPS13}, and thus so is \cqneg evaluation. In contrast, evaluation of CQs without negation is tractable for $\alpha$-acyclic queries~\cite{DBLP:conf/vldb/Yannakakis81}.

\begin{theorem}[Implicit in~\cite{DBLP:journals/tcs/OrdyniakPS13}]
  \label{thm:scq.alpha}
  \cqnegprob is \np-hard even when restricted to $\alpha$-acyclic queries.
\end{theorem}

\paragraph{Simplifying Assumptions}
To simplify the presentation we make the following assumptions on the instances of \cqnegprob.
First we assume that queries in \nswprob instances are always \emph{safe}, i.e., no variable occurs only in negative literals.
An unsafe query can always be made safe: If a variable $v$ occurs only in negative literals, we simply
add a new literal $R(v)$ with $R^D = \{(d)\mid d \in Dom\}$ to the query. The resulting query is clearly equivalent to the unsafe one on the given domain. Importantly, the additional unary literals does not change the nest-set width of the query.
At some points in the algorithm we operate on (sub)queries that are not safe. The assumption of safety is made for the starting point of the procedure.

Our second assumption is that the size of the domain is exactly a power of 2, i.e., $|Dom|=2^d$ for some integer $d$. Since we already assume safe queries, increasing the size of the domain has no effect on the solutions since the newly introduced constants cannot be part of any solution. Furthermore, this assumption increases the size of the domain at most by a constant factor less than $2$.

\subsection{Relation to Previous Work}
\label{sec:prev}

The algorithm presented here builds on the work of
Brault-Baron~\cite{DBLP:conf/csl/Brault-Baron12} for the
$\beta$-acyclic case. While we can reuse some of the main ideas, the
overall approach used there does not generalize to our setting. There the tractability is first shown for boolean domains, i.e., the domain is restricted to only two values. \cqnegprob over arbitrary domains is reduced to the problem
over the boolean domain by blowing up each variable in such a way as to encode the full domain using boolean variables.
This naturally requires every variable in the original query to be replaced by $\log_2|Dom|$ many new variables.
While this operation preserves $\beta$-acyclicity, it can increase $\nsw$ by a factor $\log_2|Dom|$.

\begin{example}
Consider the following query and a domain with $8$ elements.
\[
  q = \neg R(a, b) \land \neg S(b, c) \land \neg T(a, c)
\]
The reduction to a query over the boolean domain will then replace every variable $v$ by three variables $v_1,v_2,v_3$, resulting in the equivalent query $q_b$ over the boolean domain
\[
  \begin{aligned}
    q_b = &\neg R(a_1, a_2, a_3, b_1, b_2, b_3) \land \neg S(b_1, b_2, b_3, c_1, c_2,c_3) \land \\
    & \neg T(a_1,a_2,a_3, c_1,c_2,c_3)
\end{aligned}
\]
It is easy to see that $q$ has $\nsw(q)=2$ because any combination of two
variables is a nest-set of $q$. However, while $\{a,b\}$ is a nest-set of $q$,
this does not translate to the existence of a $2$-nest-set in $q_b$. It is easy to verify that any $\{a_i,b_j\}$ for $i,j\in[3]$
is not a nest-set. Indeed, applying the ideas from Section~\ref{sec:fpt} it is easy to see that in general, for such a triangle query, $\nsw(q_b) = 2\,\log_2|Dom|$.
\end{example}

A subtle but key observation must be made here. While the previous
example shows that the variable blowup from the binary encoding
affects the nest-set width in general, this does not happen when
$\nsw(q)=1$. Consider a nest-set $\{v\}$ of some hypergraph $H$. The edges incident to $v$ are linearly ordered by $\subseteq$. If we
 add a new vertex $v'$ in all the edges that contain $v$, then
clearly the edges incident to $v'$ are the same as those of $v$ and
therefore also linearly ordered by $\subseteq$. 

\begin{lemma}
  \label{lem:acyc}
  Let $H$ be a hypergraph with a nest-set $\{v\}$. Let $H'$ be a hypergraph obtained by adding a new variable $v'$
  to $H$ that occurs exactly in the same edges as $v$. Then $\{v\}$ and $\{v'\}$ are both nest-sets of $H'$.
\end{lemma}

This subtle difference between $1$-nest-sets and larger nest-sets will be principal to the following section.

\subsection{Eliminating Variables}
\label{sec:varelim}

The \cqnegprob algorithm in the following Section~\ref{sec:scq2} will be
based around successive elimination of variables from the query. This
elimination will be guided by a nest-set elimination ordering where we
eliminate all variables of a nest-set at once. This elimination of a
nest-set $s$ is performed in three steps.

\begin{enumerate}
\item Eliminate all occurrences of variables from $s$ in positive literals.
\item Extend the negative literals incident to $s$ in such a way that they form a $\beta$-acyclic subquery.
\item Eliminate the variables of $s$ from the $\beta$-acyclic subquery.
\end{enumerate}

In this section we introduce the mechanisms used for these steps. For steps 1 and 2 we need to extend literals in such a way that their variables include all variables from some set $s$. We do this in a straightforward way by simply extending the relation by all possible tuples for the new variables. It is then easy to see that such extensions are equivalent with respect to their set of satisfying solutions.

\begin{definition}
  Consider a literal $L(v_1,\dots,v_n)$ where $L$ is either $R$ or $\neg R$ and the respective relation $R^D$. Let $s$ be a set of variables and let $s' = s \setminus \{v_1,\dots,v_n\}$ be the variables in $s$ that are not used in the literal. We call the literal $L(v_1\dots,v_n,\vec{s'})$ with the new relation $R^D \times Dom^{|s'|}$ the \emph{$s$-extension of $R$} (where $Dom^m$ represents the $m$-ary Cartesian power of the set $Dom$ and we use the relational algebra semantics of the product $\times$ ).
\end{definition}

\begin{lemma}
  \label{lem:extension}
  Let $L'(\vec{v},\vec{s'})$ be the $s$-extension of $L(\vec{v})$. Then an assignment $a \colon vars(L) \to Dom$ satisfies $L(\vec{v})$ if and only if every extension of $a$ to $vars(L')$ satisfies $L'(\vec{v},\vec{s'})$.
\end{lemma}
\begin{proof}
  Let $L(\vec{v})$ where $L$ is either $R$ or $\neg R$ and let $L'(\vec{v},\vec{s'})$ be the $s$-extension.
  Let $a \colon vars(L) \to Dom$ be an assignment that satisfies $L$. If $L$ is positive, we have $a[\vec{v}] \in R^D$ and then every extension of the tuple to $vars(L')$ exists by the semantics of the relational product. If $L$ is negative, then $a[\vec{v}] \not\in R^D$. The relational product for creating the relation of the $s$-extension will therefore also not create any tuples where $a[\vec{v}]$ occurs in the projection to $\vec{v}$.

  On the other hand, let $a \colon vars(L) \to Dom$ such that every $a' \colon vars(L') \to Dom$ that extends $a$ satisfies $L'$.
  If $L$ is positive, then any such $a'$ also satisfies $\pi_{vars(L)}(L')$ and therefore $a$ satisfies $L$ as the relational product does not change the tuples of $L$.
  If $L$ is negative, consider the tuple $t = a[\vec{v}]$. Suppose $t \in R^D$, i.e.,  $a$ does not satisfy $L$. But then any extension of $a$ would also be in the relation of the $s$-extension, contradicting our assumption that every extension satisfies $L'$. 
\end{proof}

The process for positive elimination is simple. Straightforward
projection is used to create a positive literal without the variables
from $s$.  A new negative literal then restricts the extensions of
satisfying assignments for the new positive literal to exactly those
that satisfy the old positive literal. A slightly simpler form of this
method was already used in~\cite{DBLP:conf/csl/Brault-Baron12}.

\begin{lemma}
  \label{lem:poselim}
  Let $R(\vec{x}, \vec{s})$ be a positive literal. Define new literals $P(\vec{x})$ with $P^D = \pi_{\vec{x}}(R^D)$ and
  $\neg C(\vec{x}, \vec{s})$ with $C^D = P^D_s \setminus R^D$ where $P_s$ is the $s$-extension of $P$. Then an assignment $a$ satisfies $R(\vec{x}, \vec{s})$ if and only if $a$ satisfies $P(\vec{x}) \land \neg C(\vec{x}, \vec{s})$.
\end{lemma}
\begin{proof}
  Let $a$ be a satisfying assignment for $R(\vec{x},\vec{s})$. Clearly, $a$ also satisfies $P(\vec{x})$ and by Lemma~\ref{lem:extension} it also satisfies $P_s$. Furthermore, by construction $a[\vec{x},\vec{s}]$ is explicitly not in $C^D$ and hence $a$ also satisfies $\neg C(\vec{x},\vec{s})$.

  On the other hand. Let $a$ be a satisfying assignment for $P(\vec{x}) \land \neg C(\vec{x},\vec{s})$. By construction it is then clear that $a$ satisfies only those extensions of $\vec{x}$ that correspond to a tuple in $R^D$. At the same time $a[\vec{x}]$ is in $\pi_x(R^D)$. Hence, $a$ always corresponds to a tuple in $R^D$ and we see that $a$  satisfies $R(\vec{x},\vec{s})$.
\end{proof}

For the elimination of variables that occur only negatively we build upon a key idea from~\cite{DBLP:conf/csl/Brault-Baron12}.
There, a method for variable elimination is given for the case where the domain is specifically $\{0,1\}$. We repeat parts of the argument here to highlight some important details. Consider a query $q=\{\neg R_1, \dots, \neg R_n\}$. The main observation is that the satisfiability of the negative literals
$\neg R_1, \dots, \neg R_n$ with variables $x_1,\dots, x_m$ is equivalent to satisfiability of the formula
\[
  \bigwedge_{i=1}^n \bigwedge_{(a_{1},\dots,a_{q})\in R_i}
  \left(a_{1} \neq x_{i_1}  \lor \cdots \lor a_{q} \neq x_{i_q}\right)
\]
Since we are in the domain $\{0,1\}$ we have only two cases for the inequalities. Either $0\neq x$ or $1 \neq x$, which
are equivalent to $x=1$ and $x=0$, respectively. We can therefore equivalently rewrite the clauses in the formula above as the propositional formula
\[
  \sigma(a_1) x_{i_1}  \lor \cdots \lor \sigma(a_{q}) x_{i_q}
\]
where $\sigma(1)=\neg$ and $\sigma(0)=\epsilon$, i.e., the empty string.

Recall, if we have two clauses $x \lor \ell_1 \lor \cdots \lor \ell_\alpha$ and $\overline{x} \lor \ell'_1 \lor \ell'_\beta$ the \emph{$x$-resolvent} of the two clauses is $\ell_1 \lor \cdots \lor \ell_\alpha \lor \ell'_1 \lor \cdots \lor \ell'_\beta$.
Removing all clauses containing variable $x$ and adding all $x$-resolvents as new clauses to a given formula in CNF yields an equi-satisfiable formula without the variable $x$. This process is generally referred to as Davis-Putnam resolution~\cite{DBLP:journals/jacm/DavisP60} and its formal definition is recalled in Appendix~\ref{sec:sat}.  If we then reverse the initial transformation from query to propositional formula, we obtain a new query $q'$ (and corresponding database) that no longer contains the variable $x$. The new $q'$ has a solution if and only if $q$ has a solution.

It was already shown in~\cite{DBLP:journals/tcs/OrdyniakPS13} that resolution on a nest point will never increase the number of clauses. After conversion back to the \cqneg setting this means that every relation will contain at most as many tuples as it did before the variable elimination. Note that this conversion can be done by simply reversing the encoding as a propositional formula. Further details can be found in~\cite{DBLP:conf/csl/Brault-Baron12}. In combination with other standard properties of resolution one then arrives at the following statement.

\begin{proposition}[Implicit in Lemma~16 in~\cite{DBLP:conf/csl/Brault-Baron12}]
  \label{prop:bool}
  Let $q$ be the query $\{\neg R_1(\vec{x_1},y),\neg R_2(\vec{x_2},y),\dots, \neg R_n(\vec{x_n},y)\}$ on database $D$
  with domain $\{0,1\}$ and let $\{y\}$ be a nest-set of $q$.  There exists a query $q'$ of the form $\neg R_1(\vec{x_1}),\neg R_2(\vec{x_2}),\dots, \neg R_n(\vec{x_n})$ and a database $D'$ with the following properties:
  \begin{enumerate}
  \item If $a \in q(D)$, then $a[vars(q')] \in q'(D')$.
  \item If $a'\in q'(D')$ then there exists $c \in Dom$ such that $a' \cup \{y \mapsto c\} \in q(D)$.
  \item $q'$ and $D'$ can be computed in $O(\norm{D})$ time.%
  \item For every $i\in [n]$ we have $|R^{D'}_i| \leq |R^D_i|$
  \end{enumerate}
\end{proposition}

It was discussed in the previous section that we can not, in general, reduce a query to an equivalent query with 2 element domain without increasing the nest-set width by a $\log|Dom|$ factor. However, as mentioned in the outline above, our plan is to temporarily transform certain subqueries in such a way that they become $\beta$-acyclic and that for any variable $v$ that we want to eliminate, $\{v\}$ is a nest-set of the transformed subquery.

By the observation from Lemma~\ref{lem:acyc}, the reduction to a 2 element domain by binary encoding preserves $\beta$-acyclicity and allows us to eliminate the encoding variables of $v$ one-by-one using Proposition~\ref{prop:bool}. Afterwards, we can revert the binary encoding by mapping everything back into the original domain. This strategy allows us to lift Proposition~\ref{prop:bool} to a much more general form, allowing for variable elimination in \emph{arbitrarily large} domains. This will ultimately allow us to circumvent the obstacles described in Section~\ref{sec:prev}.

\begin{lemma}
  \label{lem:var}
  Let $q$ be the query $\{\neg R_1(\vec{x_1},y),\neg R_2(\vec{x_2},y),\dots, \neg R_n(\vec{x_n},y)\}$ on database $D$
  with $|Dom|=2^k$  and let $\{y\}$ be a nest-set of $q$.  There exists query $q'$ of the form $\neg R_1(\vec{x_1}),\neg R_2(\vec{x_2}),\dots, \neg R_n(\vec{x_n})$ and a database $D'$ with the following properties:
  \begin{enumerate}
  \item If $a \in q(D)$, then $a[vars(q')] \in q'(D')$.
  \item If $a'\in q'(D')$ then there exists $c \in Dom$ such that $a' \cup \{y \mapsto c\} \in q(D)$.
  \item $q'$ and $D'$ can be computed in $O(\norm{D}\log^2 |Dom|)$ time given $q$ and $D$ as input.
  \item For every $i\in [n]$ we have $|R^{D'}_i| \leq |R^D_i|$. %
  \end{enumerate}
\end{lemma}
\begin{proof}[Proof of Lemma~\ref{lem:var}]
  Since we have $|Dom|=2^k$ there exists a bijection
  $f : Dom \to \{0,1\}^k$ that can be efficiently computed. We then
  consider the binary version $q_b$ of $q$ where every variable $x$ is
  substituted by variables $x_1,x_2,\dots,x_k$ and the respective database $D_b$ where
  every tuple $(a_1, \dots, a_{ar(R)}) \in R^D$ becomes a $(f(a_1), \dots, f(a_{ar(R)}) \in R^{D_b}$.
  We thus clearly have that $a \in q(D)$ if and only if $f(a) \in q_b(D_b)$.

  Observe that for every $i \in [k]$ we have that $\{y_i\}$ is a nest-set of $q_b$.
  Using Proposition~\ref{prop:bool} we can then successively remove all the $\{y_i\}$ successively, each elimination requiring $O(\norm{D_b})$ time. Recall from Lemma~\ref{lem:nsw.induced.sub} that nest-sets are preserved when vertices are deleted from the hypergraph.
  Let $q'_b$ and $D'_b$ be the result of eliminating $y_i$ for every $i \in [k]$ in this fashion.

  Since exactly the substitution of $y$ was deleted we can then clearly reverse the transformation from before and
  create a $q'$ of the form $\neg R_1(\vec{x_1}),\neg R_2(\vec{x_2}),\dots, \neg R_n(\vec{x_n})$ from $q_b$ as well as the corresponding database $D'$ from $D'_b$. Again, clearly  $a \in q'(D')$ if and only if $f(a) \in q'_b(D_b)$.

  Now, if $a \in q(D)$, then $f(a) \in q_b(D_b)$. By Proposition~\ref{prop:bool} then also $f(a)[vars(q'_b)] \in q'_b(D'_b)$
  and in turn also $f^{-1}(f(a)[vars(q'_b)]) = a[vars(q')] \in q'(D')$. For the other direction,
  we proceed similarly, if $a' \in q'(D')$ then also $f(a') \in q'_b(D'_b)$. Again, by Proposition~\ref{prop:bool}
  this can be extended to an assignment $a_b \in q_b(D_b)$ and thus also implicitly to a $f^{-1}(a_b) \in q(D)$. Since 
  $f(a')$ is extended by assignments for $y_1$ through $y_k$ it follows that $f^{-1}(a_b)$ extends $a'$ by some assignment for $y$.

  Finally, the transformations to binary form and back are simple rewritings and can be done in linear time. The elimination of the $y_i$ variables requires $O(k \norm{D})$ time and we have $k = \log_2(|Dom|)$.
\end{proof}

\subsection{The Elimination Procedure}
\label{sec:scq2}

We are now ready to define our algorithm for eliminating nest-sets from CQs with negation.
Our procedure for eliminating the variables of a nest-set $s$ from a
\cqnegprob instance $q$, $D$ is described in
Algorithm~\ref{alg:selim}. Updates to the database are implicit in the algorithm. This is to be understood as adding the corresponding relation for every literal that is added to the query and removing the relations for the deleted literals.
We refer to the new instance $q'$, $D'$ returned by the algorithm as the \emph{$s$-elimination} of $q$, $D$.

\begin{algorithm}[t]
  \SetKwInput{Output}{output}
\SetKwInput{Input}{input}

  \SetKw{Halt}{Halt}\SetKw{Reject}{Reject}\SetKw{Accept}{Accept}
  \SetKwProg{Fn}{Function}{}{}
  \SetKwFunction{elimpos}{Elim-$s$-Positive}
  \SetKwFunction{elimneg}{Elim-$s$-Negative}

  \Fn{\elimpos($q$)}{
    Let $P_1, \dots, P_n$ be the positive literals incident to $s$ in $q$\;
    Let $\neg R_1, \dots, \neg R_m$ be the negative literals incident to $s$ where $vars(R_j) \subseteq \bigcup_{i=1}^n vars(P_i)$\;
    $q_J \leftarrow \{P_1, \dots, P_n, \neg R_1, \dots, \neg R_m\}$\;
    $J \leftarrow $ all solutions of $q_J(D)$\;
    $P  \leftarrow \pi_{vars(J)\setminus s}(J)$\;
    $P_s \leftarrow $ the $s$-extension of $P$\;
    $C  \leftarrow P_s \setminus J$\;

    $q_1 \leftarrow (q \setminus q_J) \cup \{P, \neg C\}$\;
    \KwRet $q_1$
  }

  \Fn{\elimneg($q_1$)}{
    Let $\neg C, \neg N_1, \dots, \neg N_\ell$ be the literals incident to $s$ in $q_1$\;
    \ForEach{$i \in \ell$}{
      $N'_i \leftarrow $ the $s$-extension of $N_i$\;
    }
    $q^\neg \leftarrow \{\neg C, \neg N'_1, \dots, \neg N'_\ell\}$\;     \label{algline:qneg}
    Let $q^* = \{\neg C^*, \neg N^*_1,\dots, \neg N^*_\ell\}$ be the query obtained by successively eliminating every $v \in s$ from $q^\neg$\ using Lemma~\ref{lem:var}\;
    $q_{-s} \leftarrow (q_1 \setminus q^\neg) \cup \{\neg N^*_1, \dots, \neg N^*_\ell\}$\;
    Update relation $P$ to $P - C^*$ in $D^{-s}$\;\label{algline:update}
    \KwRet $q_{-s}$
  }

  \Begin(\tcc*[f]{\bfseries Main}){
    $q_1 \leftarrow $ \elimpos($q$)\;
    \KwRet \elimneg($q_1$)\;
  }
\caption{Eliminate nest-set $s$ from $q, D$.}
\label{alg:selim}
\end{algorithm}

The procedure begins by eliminating all positive occurrences of $s$ via the function \texttt{Elim-$s$-Positive}.
To do so, it considers the subquery $q_J$, which contains all the positive literals incident to $s$ as well as those negative literals that are fully covered (w.r.t. their variable scopes) by these positive literals. It is straightforward to compute all the solutions of $q_J$ by first joining all the positive literals and then incorporating the negative literals via anti-joins. This can be done efficiently since the variables of $q_J$ can be covered by at most $k$ positive literals by a similar argument as in the proof of Lemma~\ref{lem:hinge}. The set of solutions of $q_J$ is taken as a new relation $J$, to which we apply the mechanism from Lemma~\ref{lem:poselim}. The resulting literals $P$ and $\neg C$ replace the subquery $q_J$ in $q$ to form $q_1$.

The resulting query $q_1$ thus is equivalent to $q$ and has no variable in $s$ occurring in a positive literal. The only literals incident to $s$ that are left are those negative literals that contain variables beyond those in $q_J$, and the new literal $\neg C$. The second subprocedure, \texttt{Elim-$s$-Negative}, eliminates the variables in $s$ from these negative literals using Lemma~\ref{lem:var}.
To eliminate a variable $v \in s$ with Lemma~\ref{lem:var} we need a $\beta$-acyclic subquery where $v$ is a nest-point. We therefore do not consider the literals $\neg N_i$ directly but instead operate on their $s$-extensions.
Observe that the literals $\neg N_i$ are all incident to $s$ and therefore their variables are linearly
ordered under $\subseteq_s$. Furthermore, for every $i \in [\ell]$ we have $vars(N_i) \setminus s \supseteq vars(C) \setminus s$. Thus, for the $s$-extensions $N'_i$ of $N_i$ we have
\[
  vars(C) \subseteq vars(N'_{i_1})  \subseteq vars(N'_{i_2}) \subseteq \cdots \subseteq vars(N'_{i_\ell})
\]
Therefore, $q^\neg$ on line~\ref{algline:qneg} in the algorithm is clearly $\beta$-acyclic and all variables in $s$ are present in every literal. Thus also every variable of $s$ is a nest-point of $q^\neg$ and Lemma~\ref{lem:var} can be used to eliminate all of them. After the elimination we get the new set of literals $q^*$ which we replace $q^\neg$ with in $q_1$. The literal $\neg C^*$ always has the same set of variables as the $P$ that was introduced in \texttt{Elim-$s$-Positive}. We thus can simply account for $\neg C^*$ by subtracting the relation of $C^*$ from the relation of $P$ instead of adding $\neg C^*$ as a literal. This way we avoid the possibility of increasing the number of new literals in the resulting final query $q_{-s}$.

\begin{example}
  \label{ex:selim}
  Consider a query $q$ with nest-set $s = \{a,b,c\}$. The query has literals $P_1(a,b,c)$, $P_2(b,d)$, $\neg N_1(a,d,e,f,g)$, and $\neg N_2(c,d,e)$ incident to $s$. This setting is illustrated on hypergraph level in Figure~\ref{fig:cnegex} where the components $C_1$ and $C_2$ abstractly represent the rest of the query.

  Algorithm~\ref{alg:selim} first computes the intermediate relation $J$ containing all the solutions for $P_1(a,b,c) \land P_2(b,d)$. We remove the variables in $s$ from $J$ by projection to obtain the new literal $P(d)$. We furthermore add $\neg C(a,b,c,d)$ as in Lemma~\ref{lem:poselim} to make $q_1$ equivalent to $q$. Note that variables from $s$ now occur only in negative literals of $q_1$.

  The procedure then moves on to eliminating $s$ from the negative
  literals. First all the negative literals are expanded to cover all
  variables of $s$. The expanded negative literals
  $\neg N'_1(a,d,e,f,g,b,c)$, $\neg N'_2(c,d,e,a,b)$, $\neg C(a,b,c,d)$
  make up the subquery $q^\neg$. As discussed above, this expansion modifies the hypergraph structure in such a way that all the variables of $s$ now correspond to $1$-nest-sets of $q^\neg$ (see also Figure~\ref{fig:cnegex}).
  They can therefore be eliminated using Lemma~\ref{lem:var} to obtain $q^*$.

  Finally, replacing $q^\neg$ in $q_1$ by $q^*$ (and simplifying $\neg C^*(d) \land P(d)$) will produce the final query without variables from $s$, the $s$-elimination of $q$.
\end{example}

\begin{figure}[t]
  \centering
  \includegraphics[width=0.95\columnwidth]{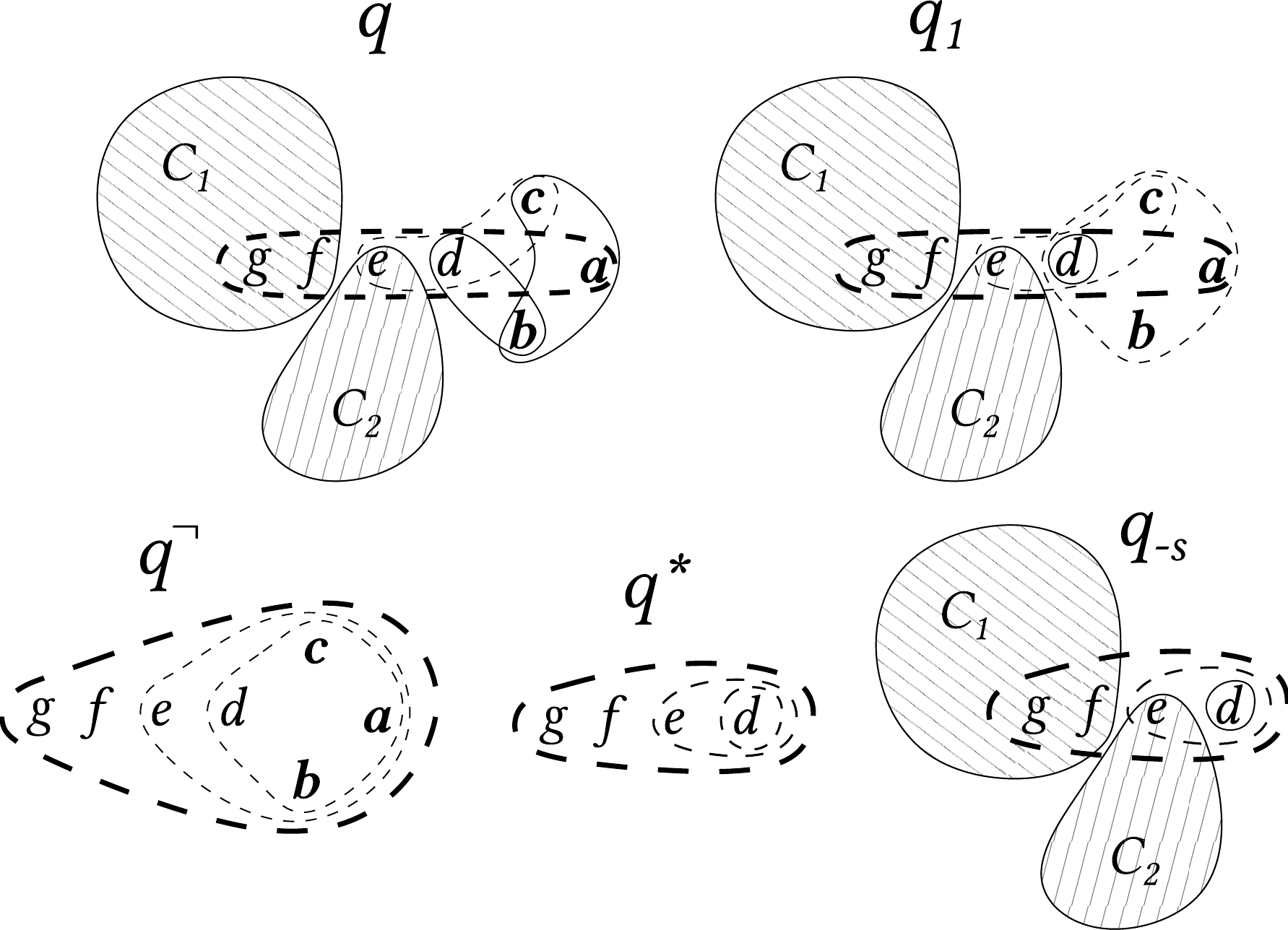}
  \Description{A visual representation of the queries described in Example~\ref{ex:selim}}
  \caption{Example of an $s$-elimination on hypergraph level. (Dashed edges correspond to negative literals.)}
  \label{fig:cnegex}
\end{figure}

What is left is to prove that an $s$-elimination has a solution if and only if the original query has a solution. This follows from the combination of the observations in Section~\ref{sec:varelim}. Moreover, the $s$-elimination is always smaller than the original query and it can be computed in polynomial time when the size of $s$ is bounded. From these three properties it will then be straightforward to establish our main result, the tractability of \cqnegprob under bounded $\nsw$. %

\begin{lemma}
  \label{lem:elim.good}
  Let $q', D'$ be the $s$-elimination of some \cqnegprob instance $q,D$. 
  Then $q'(D') \neq \emptyset$ if and only if $q(D) \neq \emptyset$.
\end{lemma}
\begin{proof}%
  From Lemma~\ref{lem:poselim} it follows that $P \land \neg C$ is equivalent to $q_J$. Therefore we have $q_1(D_1)=q(D)$.

  Suppose $a \in q_1(D_1)$. Then, for any $i\in [\ell]$ we have that $a$ satisfies $\neg N_i$. From Lemma~\ref{lem:extension} it follows that $a$ also satisfies $\neg N'_i$. Hence, $a[vars(q^\neg)]$ satisfies $q^\neg$.
  By Lemma~\ref{lem:var} it then follows that $a[vars(q^*)] \in q^*(D^*)$. Since the literals in $q^\neg$ were the only literals incident to $s$ we have that $a[vars(q)\setminus s]$ satisfies all literals in $q_{-s}$. The update of $P^{D_{-s}}$ in line~\ref{algline:update} of the algorithm is trivial since $vars(P) = vars(C^*)$.

  For the other side of the equivalence now assume that
  $a' \in q'(D')$. Then $a$ in particular satisfies all literals of
  $q^*$. By Lemma~\ref{lem:var} there exists an assignment $a_s$ to
  $s$ such that $a = (a' \cup a_s)$ satisfies $q^\neg$. Now, observe
  that for every $j \in [\ell]$, if $a[vars(N'_j)] \not \in N'_j$ then
  also $a[vars(N_j)] \not \in N_j$ by Lemma~\ref{lem:extension}.
  All other literals remain the same between $q'$ and $q_1$ and thus $a \in q_1(D) = q(D)$.
\end{proof}

\begin{lemma}
  \label{lem:elim.fast}
   Let $q,D$ be an instance of \cqnegprob and let $s$ be a $k$-nest-set of $q$. Then the $s$-elimination of $q,D$ can be computed in $O(|R_{max}|^k|Dom|^k poly(\norm{q}+\norm{D}))$
  time.
\end{lemma}
\begin{proof}%
  We first argue that $vars(q_J)$ can be covered by at most $k$ positive literals. Since $s$ is a nest-set, one of the $P_i$ is maximal (among the positive literals) with regards to $\subseteq_s$. Just like in Lemma~\ref{lem:hinge} we are left with at most $k-1$ variables of $s$ that are not covered by this maximal positive literal. As we assume safety, every such variable requires at most one positive literal to cover it and the claim holds. Thus, $J$ will contain a subset of the tuples formed by a join of $k$ positive literals, hence $|J|\leq O(|R_{max}|^k)$.

  A join $A \join B$ is feasible in $O(|A||B|\max\{ar(A),ar(B)\}\log|Dom|)$ time by a straightforward nested loop join. As $ar(R)\leq \norm{Q}$ for any relation symbol $R$ we simplify to $O(|R_{max}|^2\norm{q}\log|Dom|)$. As usual, once $k$ positive literals that cover $vars(q_J)$ are joined, any further join is simply a semi-join and requires only linear time. It follows that the joins over all the positive literals of $q_J$ can be computed in $O(|R_{max}|^k\norm{q}\log|Dom|)$ time. The negative literals in $q_J$ can then all be removed by anti-joins, which also require linear time, i.e., $A \ajoin B$ can be computed in $O(\norm{A}+\norm{B})$ time. At most $|q|$ anti-joins need to be performed and we therefore see that we have an upper bound $O(|R_{max}|^k \norm{q}^2 \log|Dom|)$ for the time required to compute $J$.

  Computing $P$ is clearly linear in $\norm{J}$ while computing $C$ is feasible
  in $O(\norm{J}\cdot \norm{Dom^k})$ time. We use the fact that $ar(J) \leq \norm{q}$ to simplify the final bound.

  For every $i \in [\ell]$, the $s$-extension from $N_i$
  to $N'_i$ requires $O(\norm{R_{max}}\cdot\norm{Dom^k})$ time. Finally, for
  each variable $v\in s$, eliminating $v$ requires
  $O(\norm{D^{\neg}}\log^2|Dom|)$ time by Lemma~\ref{lem:var} and is
  performed $k$ times.
  As the $s$-extension can increase the relation size by a factor of at most $|Dom|^k$ we have $\norm{D^\neg}\leq \norm{D}|Dom|^k$. Note that the linear factor $k$ is simplified away in the final bound by observing $k \leq vars(q)$. 
\end{proof}

\begin{lemma}
  \label{lem:elim.small}
  Let $q'$, $D'$ be the $s$-elimination of some \cqnegprob instance $q$, $D$. 
  Then $\norm{q'}\leq \norm{q}$ and $\norm{D'} \leq \norm{D}$.
\end{lemma}
\begin{proof}%
  For the query, we remove $n$ positive literals. Observe that $n\geq 1$ because we assume $q$ is safe and $s \neq \emptyset$. We only add one new positive literal $P$. Hence, the number of positive literals can not increase through $s$-elimination. Every new negative literal in $q'$ corresponds one-to-one to a negative literal that was removed from $q$. We thus have less or equal literals and strictly less variables in $q'$ than in $q$.

  For the database observe that the variables of the literals $P_1, \dots, P_n$ can be ordered as follows
  \[
    vars(P_{i_1}) \setminus s \subseteq vars(P_{i_2}) \setminus s \subseteq \cdots  \subseteq vars(P_{i_n}) \setminus s
  \]
  By construction we have that $vars(P) = \left(\bigcup_{i=1}^{n} vars(P_i)\right) \setminus s$. As the union over a chain of subsets is simply the maximal element of the chain we have $vars(P) = vars(P_{i_n})\setminus s$.
  It is then easy to see from the construction that $|P^{D'}| \leq |P_{i_n}^D|$.
  For the new negative literals we have $|N^*_i | \leq |N_i|$ by Lemma~\ref{lem:var}. Since no arities can increase in the $s$-elimination we arrive at $\norm{D'}\leq \norm{D}$.
\end{proof}

Finally, note that the construction of the $s$-elimination preserves the simplifying assumptions made in the beginning of the section. The domain is never modified by the procedure and if $q$ was safe, then so is its $s$-elimination $q'$. Moreover, we also have $H(q') = H(q) - s$ and can therefore repeatedly apply this elimination along a $k$-NEO to decide whether $q(D) \neq \emptyset$.

\begin{proof}[Proof of Theorem~\ref{thm:safe.scq}]
  Let $\mathcal{Q}$ be a class of \cqnegprob instances and say there exists a constant $k$ such that the nest-set width of every query in $\mathcal{Q}$ is at most $k$.
  
  Let $q$, $D$ be an instance of \cqnegprob with $nsw(q) \leq k$. First, we compute a $k$-NEO $\calO = (s_1,\dots, s_\ell)$ which is feasible in polynomial time for constant $k$  by Corollary~\ref{cor:fpt}.
  Then perform the following procedure that starts with $q_0 := q, D_0 := D, i:=1$:
  \begin{enumerate}
  \item If $i > \ell$, accept the input. Otherwise continue with the next step.
  \item Let $q_{i}, D_{i}$ be the $s_i$-elimination of $q_{i-1},D_{i-1}$. Rename the new $P$ literal to $P_i$.
  \item If $P_i^{D_i} = \emptyset$, reject the input. Otherwise increment $i$ by $1$ and continue from step 1.
  \end{enumerate}
  In case of acceptance,
  the procedure has eliminated all variables and only the 0-ary literal $P_\ell$ is left in $q_\ell$. Since the procedure did not reject in the step before, we have $P_\ell^{D_\ell} \neq \emptyset$, i.e., it contains the empty tuple and thus $q_\ell(D_\ell)=\{()\} \neq \emptyset$.
  On the other hand, if the procedure rejects at step $i$, then $P_i^{D_i} = \emptyset$. The literal $P_i$ occurs positively in $q_i$  and it follows that $q_i(D_i) = \emptyset$.

  By Lemma~\ref{lem:elim.good} we have $q_i(D_i) \neq \emptyset$ if and only if $q(D) \neq \emptyset$ for all $i \in [\ell]$. The described procedure is therefore sound and complete. The computation of $q_i,D_i$ from $q_{i-1},D_{i-1}$ is performed at most $\ell \leq vars(q)$ times. By Lemmas~\ref{lem:elim.small} and~\ref{lem:elim.fast} the procedure requires only polynomial time in $\norm{q}$ and $\norm{D}$.
\end{proof}

\subsection{An Application: \sat Parameterized by Nest-Set Width}
\label{sec:satfpt}

Note that the nest-set width appears only in the exponent of $|R_{max}|$ and $|Dom|$ in the time bound from Lemma~\ref{lem:elim.fast}. A reduction to \cqnegprob where these two cardinalities can be constantly bounded thus shows fixed-parameter tractability of the original problem when parameterized by $\nsw$.

Recall the reduction from \sat to \cqnegprob given in the beginning of Section~\ref{sec:scq}.
For a formula $F$ in CNF consisting of clauses $C_1,\dots,C_n$. For every clause $C_i$ with variables $v_1, \dots, v_\ell$
we create a literal $\neg R_i(v_1,\dots,v_\ell)$  where the relation $R_i^D$ contains the single tuple corresponding to the only assignment that does not satisfy the clause. To satisfy our assumption of safety we also create a positive literal $V_j(v_j)$ for every variable $v_j$ with $V^D_j = \{(0),(1)\}$, i.e., the whole domain.

We see that this reduction produces a \cqnegprob instance with $|R_{max}|=2$, and $|Dom|=2$ while $\norm{q}$ and $\norm{D}$ are linear in the size of $F$.
Plugging these values into the bound from Lemma~\ref{lem:elim.fast} gives us a $2^{O(k)}poly(|F|)$ time bound for an $s$-elimination in this query, where $k = \nsw(H(F))$ is the \emph{nest-set width of the formula $F$}. Hence, repeated $s$-elimination along a $k$-NEO gives us a fixed-parameter tractable procedure for \sat.

\begin{theorem}
  \label{thm:sat}
  \sat for propositional CNF formulas is fixed\\-parameter tractable when parameterized by the nest-set width of the formula.
\end{theorem}

Interestingly, the standard Davis-Putnam resolution procedure for \sat is also fixed-parameter polynomial when parameterized by nest-set width if resolution is performed according to a NEO. The resulting algorithm is notably different from the one induced by the above reduction and requires some further auxiliary results. We recall the details of Davis-Putnam resolution and give the alternative fixed-parameter polynomial algorithm in Appendix~\ref{sec:sat}. This serves as further evidence that $\nsw$ is a natural generalization of $\beta$-acyclicity.

Building on the discussions from Section~\ref{sec:betahw} it is also
interesting to note that \sat is known to be \textsf{W[1]}-hard when
parameterized by incidence clique
width~\cite{DBLP:journals/tcs/OrdyniakPS13}.

\section{Conclusion \& Outlook}
\label{sec:conclusion}
In this paper, we have introduced nest-set width in an
effort to generalize tractability results for $\beta$-acyclicity.  We
have established the relationship between nest-set width and related
width measures. In particular, $\nsw$ is a specialization of \bhw. In
an improvement over \bhw, checking $\nsw(H) \leq k$ is shown to be fixed-parameter
tractable when parameterized by $\nsw$. Finally, we verify that $\nsw$ is useful for generalizing tractability from $\beta$-acyclicity by proving new tractability results for boolean \cqneg evaluation and \sat.

The possibilities for future work are plentiful. With any new island of tractability, there comes a question of whether the result can be generalized further or if this is the limit of tractability for the problem. Both kinds of answers would be of great interest for \cqneg evaluation as well as \sat.

An interesting question has been left open in this paper: the relationship between nest-set width and point-width.
Since the tractable computation of point-decompositions remains an open question, the applicability of point-width for algorithmic results in our setting is not clearly established.
Recall, proving the relationship of point-width to $\beta$-acyclicity in~\cite{DBLP:conf/lics/Carbonnel0Z19} already required considerable effort and it is therefore likely that showing the relationship to nest-set width will be even more challenging and requires individual study. Nonetheless, we believe this to be an important question for the overall program of $\beta$-acyclicity generalizations.

Finally, our results make us hopeful that $\nsw$ can find broader application beyond the problems tackled in this paper.
For example, the tractability of \#\sat for $\beta$-acyclic formulas from~\cite{DBLP:conf/stacs/Brault-BaronCM15} is based on nest point elimination and its generalization is thus a natural candidate for further investigation. Another promising avenue of research is the application to the worst-case analysis for $\beta$-acyclic queries presented in~\cite{DBLP:conf/pods/NgoNRR14}.

\begin{acks}
  The author is very grateful to Reinhard Pichler for his valuable comments.
This work was supported by the Austrian Science Fund (FWF) project P30930, the Royal Society “RAISON DATA” project  (Reference No. RP\textbackslash{}R1\textbackslash{}201074), and the VADA (Value Added Data Systems, EP/M025268/) extension project by the University of Oxford.
\end{acks}

\bibliographystyle{ACM-Reference-Format}
\bibliography{pods21}

\appendix

\section{Nest-Set Width \& Davis-Putnam Resolution}
\label{sec:sat}
A particularly important problem where the restriction to $\beta$-acyclic instances leads to tractability is the propositional satisfiability problem (\sat)~\cite{DBLP:journals/tcs/OrdyniakPS13}. In order for formulas to relate directly to hypergraphs we consider only propositional formulas in conjunctive normal form (CNF). The hypergraph $H(F)$ of a formula $F$ has as its vertices the variables of $F$ and  edges $E(H(F)) = \{vars(C) \mid C \mbox{ clause in } F\}$ .
We then alternatively refer to $\nsw(H(F))$ as the nest-set width $\nsw(F)$ of the propositional CNF formula $F$. In this section we show that \sat is fixed-parameter tractable when parameterized by the nest-set width of the formula.

Building on the approach used in~\cite{DBLP:journals/tcs/OrdyniakPS13}, we will show that Davis-Putnam (DP) resolution~\cite{DBLP:journals/jacm/DavisP60} on a nest-set
will always decrease the number of clauses in the formula. While the number of clauses can increase in the intermediate steps, before resolution on every variable from the nest-set has been performed, the intermediate blowup can be bounded in the size of the nest-set.

\subsection{Resolution}

For this section we consider a \emph{clause} $C$ as a set of literals $x$ or $\overline{x}$ where $x$ is a variable.
If a literal is of the form $x$ we say it is positive and otherwise it is negative. This is also referred to as the \emph{phase} of the literal. If $C$ is a clause, we write $\overline{C}$ for the clause $\{\overline{\ell}\mid \ell \in C\}$, i.e., every clause is switched (note that $\overline{\overline{x}} = x$).
A formula $F$ (in CNF) is a set of clauses.

For two clauses $C,D \in F$ with $C \cap \overline{D} = \{x\}$
we call $(C \cup D) \setminus \{x,\overline{x}\}$ the \emph{$x$-resolvent} of $C$ and $D$. Note that we don't need to care
about cases where $C \cap \overline{D} \supset \{x\}$. Say the intersection equals $\{x,y\}$. Resolving on $x$ would yield
a new clause containing both $y$ and $\overline{y}$. Such a new clause is therefore trivially satisfied and of no interest.

Let $DP_x(F)$ be the formula obtained by first adding all
$x$-resolvents to $F$ and then removing all clauses where $x$ occurs.
Such a step is commonly called a \emph{(Davis-Putnam) resolution
  step}~\cite{DBLP:journals/jacm/DavisP60}. It is well-known that $F$
and $DP_x(F)$ are equisatisfiable for all variables $x$. %

We write
$DP_{x_1,\dots,x_q}(F)$ for
$DP_{x_q}(DP_{x_{q-1}}(\cdots DP_{x_1}(F)\cdots))$.  We also write
$DP_s$ for a set $s \subseteq vars(F)$ if the particular order does
not matter.  The procedure $DP_{vars(F)}(F)$ will produce either an
empty clause at some step or end in an empty formula.  In the first
case, $F$ is unsatisfiable, and conversely, $F$ is satisfiable in the
second case.

The hypergraph $H(F)$ of a formula $F$ has vertex set $vars(F)$ and edges $\{vars(C) \mid C \in F\}$.
For a set of variables $s=\{v_1,\dots,v_q\}$ we define $\overline{s}:= \{\overline{v_1},\dots,\overline{v_q}\}$.
Let $C$ be a clause and $s$ a set of variables. We write $C-s$ for $C \setminus (s \cup \overline{s})$, the clause with all literals of variables from $s$ removed. We extend this notation to formulas as
\(
  F-s := \{C -s\mid C \in F\}
\)
for formula $F$.

\subsection{Davis-Putnam Resolution over Nest-Sets}

In general, $DP_x(F)$ can contain more clauses than $F$ and the whole
procedure can therefore require exponential time (and space). However,
as we will see, if we eliminate nest-sets then the increase in clauses is only temporary. We start by showing that resolution on a variable in a nest-set will only produce resolvents that remain, in a sense, local to the nest-set.

\begin{lemma}
  \label{lem:nestres}
  Let $s=\{v_1, \dots, v_q\}$ be a nest-set of $H(F)$. For every $s' \subseteq s$
  and any clauses  $C,D \in DP_{s'}(F)$ that contain a variable from $s$ we have that $vars(C) \subseteq_s vars(D)$,
  or vice versa.
\end{lemma}
\begin{proof}
  Proof is by induction on the cardinality of $s'$. If $s'=\emptyset$ no resolution takes place and the statement follows from the fact that $s$ is a nest-set.

  Say $s'= \{v_1, \dots, v_k\}$ and consider
  $C,D \in DP_{v_1,\dots,v_{k-1}}(F)$ such that both clauses contain a
  variable from $s$. According to the induction hypothesis, w.l.o.g.,
  we have $vars(C) \subseteq_s vars(D)$ (the other case is symmetric). To show that the statement also holds for $s'$
  we show that each $v_k$-resolvent of $DP_{v_1,\dots,v_{k-1}}(F)$ also satisfies this property.

  In particular, for such $C,D$ and $C \cap \overline{D} = \{v_k\}$ we show that $vars(R) \setminus s = vars(D)\setminus s$, where $R$ is the $v_k$-resolvent of $C$ and $D$. Since $vars(D) \setminus r$ is comparable to all other clauses that contain a variable of $s$, so is $vars(R)\setminus s$ and the statement also holds for $DP_{s'}(F)$.

  It is not hard to see that the two sets are in fact the same:
  Since $R = (C \cup D) \setminus \{v_k, \overline{v_k}\}$ and $v_k\in s$ we also have $var(R) \setminus s= (var(C) \cup var(D)) \setminus s$. Recall that we have $(var(C) \setminus s) \subseteq (var(D) \setminus s)$
  and therefore $(var(C) \cup var(D)) \setminus s = var(D) \setminus s$.
  Note that the argument doesn't change if $D \cap \overline{C} = \{v_k\}$.
\end{proof}

\begin{lemma}
  \label{lem:dpres.minus}
  Let $s=\{v_1, \dots, v_q\}$ be a nest-set of $H(F)$. For every $s' \subseteq s$
  and any clause  $C \in DP_{s'}(F)$  we have that $C-s \in F-s$.
\end{lemma}
\begin{proof}
  Proof is by induction on the cardinality of $s'$. For $s'=\emptyset$
  the statement is true by definition.

  Suppose the statement holds for $|s'|<k$ and let
  $s'= \{v_1, \dots, v_k\}$. Consider a
  $C \in DP_{s'}(F) \setminus DP_{s'\setminus \{v_k\}}(F)$, i.e., a
  new clause obtained by the resolution on $v_k$ after resolution on
  all the other variables of $s'$ was already performed.  Thus,
  $C = (C_1 \cup C_2) \setminus \{v_k, \overline{v_k}\}$ for some
  $C_1, C_2 \in DP_{s'\setminus \{v_k\}}(F)$ where $C_1 \cap \overline{C_2} = \{v_k\}$.  By
  Lemma~\ref{lem:nestres} we have, w.l.o.g.,
  $vars(C_1) \setminus s \subseteq vars(C_2) \setminus s$ (the other case is symmetric).

  Now, for every variable
  $v \in (vars(C_1) \cap vars(C_2)) \setminus s$, we know that $v$
  occurs in the same phase in both clauses since
  $C_1 \cap \overline{C_2} = \{v_k\}$. Thus, also $C_1-s \subseteq C_2 - s$
  and therefore also $C-s = C_2 -s$. By the induction
  hypothesis we have $C_2-s \in F-s$ and the proof is complete.
\end{proof}

While Lemma~\ref{lem:dpres.minus} has $DP_s(F)\leq |F|$ as a direct consequence, it also gives insight into the size of the intermediate formulas. In particular, for any non-empty $s' \subseteq s$ we have $|DP_{s'}(F)| \leq 3^{k-|s'|}|F|$.
This can be observed from noting that any clause $C \in  DP_{s'}(F)$ is an extension of a clause from $F-s$ by any combination of literals for the variables $s \setminus s'$. Specifically, there are three possibilities for every such variable, it either occurs positively, negatively or not at all in $C$. Thus, any clause in $F-s$ has only $3^{|s\setminus s'|} = 3^{k-|s'|}$ extensions.

From Section~\ref{sec:fpt} we know that we can compute a $k$-NEO of $H(F)$ in fixed-parameter polynomial time when parameterized by $k$. From the size bound above it is then easy to see that each resolution step in a nest-set can be performed in fixed-parameter polynomial time. Hence, repeating the resolution step along a $k$-NEO is an fixed-parameter tractable procedure for \sat parameterized by $k$.

\begin{theorem}
  \label{thm:sat.dpres}
  \sat for propositional CNF formulas is fixed-parameter tractable when parameterized by the nest-set width of the formula.
\end{theorem}

\end{document}